\documentclass[]{IEEEtran}

\usepackage{mystyle-ieee}
\usepackage{epstopdf}

\begin{document}

\title{Robust Sparse Analysis Regularization}

\author{Samuel~Vaiter,
        Gabriel~Peyr\'{e},
        Charles~Dossal
        and~Jalal Fadili
\thanks{S. Vaiter and G. Peyr\'{e} are with CNRS and CEREMADE, Universit\'{e} Paris-Dauphine, Place du Mar\'{e}chal De Lattre De Tassigny, 75775 Paris Cedex 16, France, email: \texttt{samuel.vaiter@ceremade.dauphine.fr}}
\thanks{C. Dossal is with IMB, Universit\'{e} Bordeaux 1, 351, Cours de la lib\'{e}ration, 33405 Talence Cedex, France.}
\thanks{Jalal Fadili is with GREYC, CNRS-ENSICAEN-Universit\'{e} de Caen, 6, Bd du Mar\'{e}chal Juin, 14050 Caen Cedex, France.}}

\maketitle

\begin{abstract}
  This paper investigates the theoretical guarantees of $\lun$-analysis regularization when solving linear inverse problems. Most of previous works in the literature have mainly focused on the sparse synthesis prior where the sparsity is measured as the $\lun$ norm of the coefficients that synthesize the signal from a given dictionary. In contrast, the more general analysis regularization minimizes the $\lun$ norm of the correlations between the signal and the atoms in the dictionary, where these correlations define the analysis support. The corresponding variational problem encompasses several well-known regularizations such as the discrete total variation and the Fused Lasso.

Our main contributions consist in deriving sufficient conditions that guarantee exact or partial analysis support recovery of the true signal in presence of noise. More precisely, we give a sufficient condition to ensure that a signal is the unique solution of the $\lun$-analysis regularization in the noiseless case. The same condition also guarantees exact analysis support recovery and $\ldeux$-robustness of the $\lun$-analysis minimizer vis-\`a-vis an enough small noise in the measurements. This condition turns to be sharp for the robustness of the analysis support. To show partial support recovery and $\ldeux$-robustness to an arbitrary bounded noise, we introduce a stronger sufficient condition. When specialized to the $\lun$-synthesis regularization, our results recover some corresponding recovery and robustness guarantees previously known in the literature. From this perspective, our work is a generalization of these results. We finally illustrate these theoretical findings on several examples to study the robustness of the 1-D total variation and Fused Lasso regularizations.
\end{abstract}

\begin{IEEEkeywords}
sparsity, analysis regularization, synthesis regularization, inverse problems, $\lun$ minimization, union of subspaces, noise robustness, total variation, wavelets, Fused Lasso.
\end{IEEEkeywords}
\section{Introduction} 
\label{sec:intro}

\subsection{Inverse Problems and Signal Priors} 
\label{sub:ip}

This paper considers the stability of regularized inverse problems using sparsity-promoting priors.
The forward model in many data acquisition scenarios can be formulated as the action of a linear mapping on some unknown (sought-after) signal contaminated by an additive noise. This takes the form
\begin{equation}\label{eq:ip}
  y = \Phi x_0 + w,
\end{equation}
where $y \in \RR^Q$ are the observations, $x_0 \in \RR^N$ the unknown signal to recover, $w$ the noise supposed to be of bounded $\ldeux$-norm, and $\Phi$ a bounded linear operator which maps the signal domain $\RR^N$ into the observation domain $\RR^Q$ where generally $Q \leq N$.
Even when $Q=N$, the mapping $\Phi$ is in general ill-conditioned or even singular. This makes the problem of solving for an accurate approximation of $x_0$ from the system \eqref{eq:ip} ill-posed, see for instance~\cite{kirsch2011introduction} for an introduction to inverse problems.

However, the situation radically changes if one has some prior information about the underlying object $x_0$. Regularization is a popular way to impose such a prior, hence making the search for solutions feasible. The general variational problem we consider can be stated as
\begin{equation}\label{eq:reg-noise}
  \umin{x \in \RR^N} \dfrac{1}{2} \norm{y - \Phi x}_2^2 + \lambda \J(x) ,
\end{equation}
where the first term is the data fidelity reflecting $\ldeux$-boundedness of the noise, and $\J$ is an appropriate (prior) regularization term through which some regularity is enforced on the recovered signal. The regularization parameter $\la>0$ should be adapted to balance between the allowed fraction of noise level and regularity as dictated by the prior on $x_0$.

For noiseless observations, i.e. $w=0$, taking the limit $\lambda \to 0$, we end up solving the constrained problem
\begin{equation}\label{eq:reg-noiseless}
  \umin{x \in \RR^N} \J(x) \qsubjq \Phi x = y .
\end{equation}

A popular class of priors are quadratic forms $\J(x) = \dotp{x}{Kx}$ where $K$ is a symmetric semidefinite positive kernel.
Problems \eqref{eq:reg-noise} and \eqref{eq:reg-noiseless} then correspond to Tikhonov regularization which typically induces some kind of uniform smoothness in the recovered signal. More advanced priors that have received considerable interest in the recent years rely on non-quadratic, generally nonsmooth, functionals such as those promoting sparsity of the signal in some transform domain (e.g. its wavelet transform or its derivatives). These sparsity priors are at the heart of this paper. They will be discussed in more detail after some necessary definitions and notations are first introduced in the following section.


\subsection{Notations} 
\label{sub:notations}

Throughout the paper, we focus on real vector spaces. The variable $x$ will denote a vector in $\RR^N$, $y$ will be a vector in $\RR^Q$ and $\alpha$ a vector in $\RR^P$.

The sign vector $\sign(\alpha)$ of $\alpha \in \RR^P$ is
\begin{equation*}
  \forall i \in \ens{1, \dots, P}, \quad
  \sign(\alpha)_i =
  \begin{cases}
    + 1 & \qifq \alpha_i > 0,\\
    0 & \qifq \alpha_i = 0, \\
    - 1 & \qifq \alpha_i < 0.
  \end{cases}
\end{equation*}
It suport is
\begin{equation*}
  \supp(\alpha) = \enscond{i \in \ens{1, \dots, P}}{\alpha_i \neq 0} .
\end{equation*}
For a set $I \in E$, $\abs{I}$ will denote its cardinality, and $I^c=E \setminus I$ its complement.

The $p, q$-operator (induced) norm of a matrix $M$ is
\begin{equation*}
  \norm{M}_{p, q} = \umax{x \neq 0} \frac{\norm{Mx}_q}{\norm{x}_{p}} .
\end{equation*}

The matrix $M_J$ for $J$ a subset of $\ens{1, \dots, P}$ is the submatrix whose columns are indexed by $J$.
Similarly, the vector $s_J$ is the restriction of $s$ to the entries of $s$ indexed by $J$.

The matrix $\Id$ is the identity matrix, where the underlying space will be clear from the context.
For any matrix $M$, $M^+$ is the Moore--Penrose pseudoinverse of $M$ and $M^*$ is the adjoint matrix of $M$.
$M^{+,*}$ is the adjoint of the Moore--Penrose pseudoinverse of $M$.


\subsection{Synthesis and Analysis Sparsity Priors} 
\label{sub:avs}

\paragraph{Synthesis sparsity prior} 
\label{par:synthesis}

Sparse regularization is a popular class of priors to model natural signals and images, see for instance~\cite{mallat2009a-wav}.
In its simplest form, the sparsity of coefficients $\alpha \in \RR^P$ is measured using the $\lzero$ pseudo-norm
\begin{equation*}
  \J_0(\alpha) = \norm{\alpha}_0 = \abs{ \supp (\alpha) } .
\end{equation*}
Minimizing \eqref{eq:reg-noise} or \eqref{eq:reg-noiseless} with $\J = \J_0$ is however known to be NP-hard, see for instance~\cite{natarajan1995sparse}.
Several workarounds have been proposed to alleviate this difficulty. A first family of methods relies on greedy algorithms~\cite{needell2008greedy}.
The most popular ones are Matching Pursuit~\cite{mallat1993matching} and Orthogonal Matching Pursuit~\cite{pati1993orthogonal,davis94adaptive}.
A second family of methods, which is the focus of this paper, relies on convex relaxation which amounts to replacing the $\lzero$ pseudo-norm by the $\lun$ norm~\cite{donoho2006most}.

A dictionary $D=(d_i)_{i=1}^{P}$ is a (possibly redundant, i.e. $P > N$) collection of $P$ atoms $d_i \in \RR^N$.
It can also be viewed as a linear mapping from $\RR^P$ to $\RR^N$ which is used to synthesize a signal $x \in \Im(D) \subseteq \RR^N$ as
\begin{equation*}
  x = D \alpha = \sum_{i=1}^P \alpha_i d_i ,
\end{equation*}
where $\alpha$ is the coefficient vector that synthesizes $x$ from the dictionary $D$.
The sparsest set of coefficients, according to the $\lun$ norm, defines a prior
\begin{equation*}
  \J_S(x) = \min_{\alpha \in \RR^P} \normu{\alpha} \qsubjq x = D \alpha.
\end{equation*}
Therefore any solution $x$ of \eqref{eq:reg-noise} using $\J=\J_S$ can be written as $x=D \alpha$ where $\alpha$ is a solution of
\begin{equation}\label{eq:lasso-s}
  \umin{\alpha \in \RR^P} \dfrac{1}{2} \norm{y - \Psi \alpha}_2^2 + \lambda \normu{\alpha} ,
\end{equation}
where $\Psi = \Phi D$.
$\lun$ regularization was first considered in the statistical community in \cite{tibshirani1996regre} where it was coined Lasso. Note that it was originally introduced as an $\lun$-ball constrained optimization and in the overdetermined case.
It is also known in the signal processing community as Basis Pursuit DeNoising~\cite{chen1999atomi}.
Such a problem corresponds to the so-called sparse synthesis regularization as sparsity is assumed on the coefficients $\alpha$ that synthesize the signal $x = D \alpha$.
In the noiseless case, the constrained problem \eqref{eq:reg-noiseless} becomes
\begin{equation}\label{eq:bp-s}
  \umin{\alpha \in \RR^P} \normu{\alpha} \qsubjq y = \Psi \alpha ,
\end{equation}
which goes by the name of Basis Pursuit after \cite{chen1999atomi}.
Taking $D = \Id$ amounts to assuming sparsity of the signal itself, and was used for instance for sparse spike train deconvolution in seismic imaging~\cite{santosa1986linear}.
Sparsity in orthogonal as well as redundant wavelet dictionaries are popular to model natural signals and images that exhibit certain singularities~\cite{mallat2009a-wav}.


\paragraph{Analysis sparsity prior} 
\label{par:analysis}

Analysis regularization corresponds to using $\J=\J_A$ in \eqref{eq:reg-noise} where
\begin{equation*}
  \J_A(x) = \normu{D^* x} = \sum_{i=1}^P \abs{\dotp{d_i}{x}}
\end{equation*}
which leads to the following minimization problem
\begin{equation}\label{eq:lasso-a}\tag{\lassotag}
  \umin{x \in \RR^N} \minform .
\end{equation}
Of course, $D^*$ is not in general the adjoint operator of a full rank dictionary $D$. Note that the analysis problem \eqref{eq:lasso-a} is more general than the synthesis one \eqref{eq:lasso-s} because the latter is recovered by taking $D=\Id$ and $\Psi=\Phi$ in the former.

As the objective in \eqref{eq:lasso-a} is proper (i.e. not infinite everywhere), continuous and convex, it is a classical existence result that the set of (global) minimizers is nonempty and compact if and only if
\begin{equation}\label{eq:H0}\tag{$H_0$}
  \Ker \Phi \cap \Ker D^*= \ens{0} .
\end{equation}
From now on, we suppose that this condition holds.

In the noiseless case, the $\lun$-analysis equality-constrained problem is
\begin{equation}\label{eq:bp-a}\tag{\bptag}
  \umin{x \in \RR^N} \normu{D^* x} \qsubjq \Phi x = y .
\end{equation}

One of the most popular analysis sparsity-inducing regularizations is the total variation, which was first introduced for denoising (in a continuous setting) in~\cite{rudin1992nonlinear}. It roughly corresponds to taking $D^*$ as a derivative operator.
Typically, for 1-D discrete signals, $D$ can be taken as a dictionary of forward finite differences $D_{\DIF}$ where
\begin{equation}\label{eq-d-diff}
  D_{\DIF} = 
  \begin{pmatrix}
    -1 & & \multicolumn{2}{c}{\text{\kern0.5em\smash{\raisebox{-1ex}{\Large 0}}}} \\
    +1 & -1 & \\
    & +1 & \ddots  & \\
    & & \ddots & -1 \\
    \multicolumn{2}{c}{\text{\kern-0.5em\smash{\raisebox{0.75ex}{\Large 0}}}} & & +1
  \end{pmatrix} .
\end{equation}
The corresponding prior $\J_A$ favors piecewise constant signals and images.
A comprehensive review of total variation regularization can be found in~\cite{chanbolle2009tv}.

The theoretical properties of total variation regularization have been previously studied.
A distinctive feature of this regularization is its tendency to yield a staircasing effect, where discontinuities not present in the original data might be artificially created by the regularization.
This effect has been studied by Nikolova in the discrete case in a series of papers, see e.g. \cite{nikolova2000local}, and in \cite{Ring2000} in the continuous setting. The stability of the discontinuity set of the solution of the 2-D continuous total variation based denoising problem is investigated in~\cite{caselles2008discontinuity}. Section \ref{sub:tv} shows how our results also shed some light on this staircasing effect for 1-D discrete signals.

It is also possible to use a dictionary $D$ of translation invariant wavelets, so that the corresponding regularization term $\J_A$ can be viewed as a multiscale (higher order) total variation \cite{steidl2005equivalence}. Such a prior tends to favor piecewise regular signals and images.
From a numerical standpoint, an extensive study is reported in is reported in~\cite{selesnick2009signal} using these redundant dictionaries to highlight differences between synthesis and analysis sparsity priors for inverse problems.

As a last example of sparse analysis regularization, we would like to mention the Fused Lasso~\cite{tibshirani2005sparsity}, where $D$ is the concatenation of a discrete derivative and a weighted identity. The corresponding prior $\J_A$ promotes both sparsity of the signal and its derivative, hence favoring the grouping of non-zero coefficients in blocks.


\paragraph{Synthesis versus analysis priors} 
\label{par:avs}

In a synthesis prior, the vector $\alpha$ that synthesizes the signal $x$ from the dictionary $D$ is sparse, whereas in an analysis prior, the correlation between the signal $x$ and the atoms in the $D$ is sparse. Some insights on the relation and distinction between analysis and synthesis-based sparsity regularizations were first given in \cite{elad2007analy}. When $D$ is orthogonal, and more generally when $D$ is square and invertible, \lasso and the Lasso entail equivalent regularizations in the sense that the set of minimizers of one problem can be retrieved from that of an equivalent form of the other  through a bijective change of variable. However, when $D$ is redundant, synthesis and analysis regularizations differ significantly.



\subsection{Union of Subspaces Model} 
\label{sub:union}

As analysis regularization involves the sparsity of the correlation vector $D^* x$, it is thus natural to keep track of the support of $D^* x$. To fix terminology, we define this support and its complement.

\begin{defn}\label{defn:dsup}
  The \emph{\dsup} $I$ of a vector $x \in \RR^N$ is $I =\supp(D^* x) \subset \ens{1,\dots,P}$.
  Its \emph{\dcosup}$J$ is $J = I^c=\ens{1,\dots,P} \setminus I$.
\end{defn}

A signal $x$ such that $D^* x$ is sparse lives in a subspace $\GJ$ of small dimension whose formal definition is as follows. 
\begin{defn}
  Given a dictionary $D$, and $J$ a subset of $\ens{1,\dots,P}$, the \emph{cospace} $\GJ$ is defined as
  \begin{equation*}
    \GJ = \Ker D_J^* ,
  \end{equation*}
  where we recall that $D_J$ is the subdictionary whose columns are indexed by $J$.
\end{defn}
Following the cosparse model introduced in \cite{nam2011the-c}, the signal space can thus be decomposed as 
\begin{equation*}
  \RR^N = \bigcup_{k \in \ens{0, \dots, N}}\Theta_k,
\end{equation*}
where
\begin{equation}\label{eq:subspaces}
  \Theta_k = \enscond{\Gg_J}{J \subseteq \ens{1,\dots,P} \text{ and } \dim \Gg_J = k} ,
\end{equation}
which is dubbed \emph{union of subspaces} of dimension $k$.\\

The union of subspaces associated to synthesis regularization, i.e. $D=\Id$, corresponds to $\Theta_k$ as the set of axis-aligned subspaces of dimension $k$.
For the 1-D total variation prior, where $D=D_{\DIF}$ as defined in \eqref{eq-d-diff}, $\Theta_k$ is the set of piecewise constant signals with $k-1$ steps.
Several examples of subspaces $\Theta_k$, including those corresponding to translation invariant wavelets, are discussed in~\cite{nam2011the-c}.

More general union of subspaces models (not necessarily corresponding to analysis regularizations) have been introduced in sampling theory to model various types of non-linear signal ensembles, see for instance~\cite{lu2008a-the}.
Union of subspaces models have been extensively studied for the recovery from pointwise sampling measurements~\cite{lu2008a-the} and compressed sensing measurements~\cite{eldar2009robust, baraniuk2010model, blumensath2009sampling, boufounos2011sparse}.


\subsection{Organization of this Paper} 
\label{sub:organization}
The rest of the paper is organized as follows.
Section \ref{sec:contrib} details our main contributions.
Section \ref{sec:related} draws some connections with relevant previous work.
Section \ref{sec:examples} illustrates our results in some examples.
The proofs are deferred to Section \ref{sec:proofs}.



\section{Contributions} 
\label{sec:contrib}

This paper proves the following three main results:
\begin{enumerate}
  \item \textbf{Robustness to small noise:} we provide a sufficient condition on $x_0$ ensuring that the solution of \lasso is unique, lives in the same cospace and close to $x_0$ when $w$ is small enough.

  \item \textbf{Noiseless identifiability:} under the same sufficient condition, $x_0$ is guaranteed to be the unique solution of \bp when $w=0$. 

  \item \textbf{Robustness to bounded noise:} we then give a sufficient condition that depends on the \dcosup of $x_0$ under which the solution of \lasso is unique and close to $x_0$ for an arbitrary bounded noise $w$, with the proviso that $\la$ is large enough.
\end{enumerate}
Each contribution will be rigorously described in a corresponding subsection.\\

It is worth mentioning that our results will extend previously known ones in the synthesis case, see for instance~\cite{dossal2007recovery,fuchs2004on-sp,tropp2006just-,tropp2004greed,candes2006robust}.
Additionally, there are only a few recent works that we aware of and which give provable guarantees using analysis regularization for exact recovery in the noiseless case \cite{nam2011the-c}, or accurate and robust recovery in the noisy case \cite{candes2010compr,giryes2012iterative,peleg2012cosparse,grasmair2011linear,burger2004convergence,NeedelWard12}. We will discuss this prior literature in detail in Section~\ref{sec:related}. Nevertheless, to the best of our knowledge, it appears that our work is the first that addresses the above three questions in the analysis case. \\

For some cosupport $J$, the invertibility of $\Phi$ on $\GJ$ will play a pivotal role in our theory.
This is achieved by imposing that
\begin{equation}\label{eq:hj}\tag{$H_J$}
  \Ker \Phi \cap \GJ = \ens{0}.
\end{equation}
To get the gist of the importance of \eqref{eq:hj}, consider the noiseless case where we want to recover a $D$-sparse signal $x_0$ from $y=\Phi x_0$. Let $J$ be the \dcosup of $x_0$ and assume that it is known. As $x_0 \in \GJ \cap \{x: y = \Phi x\}$, for $x_0$ to be uniquely recovered from $y$, \eqref{eq:hj} must be verified. Conversely, if $x_0$ is such that \eqref{eq:hj} does not hold, then any $x_0+h$, with $h \in \Ker \Phi \cap \GJ$, is also a candidate solution, i.e. $x_0+h \in \GJ \cap \{x: y = \Phi x\}$. Clearly, one cannot reconstruct such $D$-sparse objects.

With assumption \eqref{eq:hj} at hand, we are in position to define the following matrix whose role will be clarified shortly.
\begin{defn}
  Let $J$ be a \dcosup.
  Suppose that \eqref{eq:hj} holds.
  We define the operator $\AJ$ as
  \begin{equation}\label{eq:aj}
    \AJ = \UJ \pa{\UJ^* \Phi^* \Phi \UJ}^{-1} \UJ^* .
  \end{equation}
  where $\UJ$ is a matrix whose columns form a basis of $\GJ$. 
\end{defn}
The operator $\AJ$ can be computed without an explicit basis of $\GJ$ as an optimization problem
\begin{equation*}
  \AJ u = \uargmin{D_J^* x = 0} \frac{1}{2}\norm{\Phi x}^2 - \dotp{x}{u} .
\end{equation*}

\subsection{Robustness to Small Noise} 
\label{sub:contrib-small}

Our first contribution consists in showing that $\lun$-analysis regularization is robust to a small enough noise under a sufficient condition that depends on the sign of $D^* x_0$ and its \dcosup. This condition will be formulated via the following criterion.

\begin{defn}\label{def:asc}
	Let $s \in \{-1, 0, +1\}^P$, $I$ its support and $J = I^c$.
  Suppose that \eqref{eq:hj} holds.
  The analysis \emph{Identifiabiltiy Criterion} $\IC$ of $s$ is defined as
	\begin{equation*}
    \IC(s) = \umin{u \in \Ker D_J}
    \norm{\CJ^{[J]} s_I - u}_\infty
  \end{equation*}
  where
  \begin{equation*}
    \CJ^{[J]} = D_J^+(\Phi^* \Phi \AJ - \Id)D_I .
  \end{equation*}
\end{defn}

We have the following theorem.
\begin{thm}\label{thm:small}
  Let $x_0 \in \RR^N$ be a fixed vector of \dsup $I$ and \dcosup $J = I^c$. Let $y=\Phi x_0+w$.
  Assume that \eqref{eq:hj} holds and $\IC(\signxx) < 1$.
  Then there exist constants $c_J > 0$ and ${\tilde{c}}_J > 0$ satisfying
  \begin{equation*}
    \frac{\norm{w}_2}{T} < \frac{\tilde{c}_J}{c_J} \qandq T = \umin{i \in \ens{1,\cdots,\abs{I}}} \abs{D_I^* x_0}_i,
  \end{equation*}
  such that if $\lambda$ is chosen according to
  \begin{equation*}
    c_J \norm{w}_2 < \lambda < T{\tilde{c}}_J ,
  \end{equation*}
  the vector
  \begin{equation}\label{eq:sol-small}
    \xsolyp = x_0 + \AJ \Phi^* w - \lambda \AJ D_I \signxx ,
  \end{equation}
  is the unique solution of \lasso.
  Moreover,
  \begin{equation*}
    \xsolyp \in \GJ \qandq \signxx = \sign(D^*_I \xsolyp) .
  \end{equation*}
\end{thm}
In plain words, Theorem~\ref{thm:small} asserts that when $\IC(\signxx) < 1$, the support and sign of $D^*x_0$ are exactly recovered by solving \lasso with $\la$ wisely chosen and provided that the non-zero entries of $D_I^* x_0$ are large enough compared to noise.
In addition, if $\la$ is chosen proportional to the noise level, \eqref{eq:sol-small} gives
\begin{equation*}
  \norm{\xsolyp - x_0}_2 = O(\norm{w}_2) .
\end{equation*}

\begin{rem}
One may question the benefit of minimizing over $\Ker D_J$ in the criterion $\IC$. First note that $\IC(s)$ is upper-bounded by $\norm{\CJ^{[J]} s_I}_\infty$. For $D$ with maximally linear independent columns, $\Ker D_J=\{0\}$ and $\IC$ is large. On the other hand, when $\Ker D_J $ is large, minimizing the (translated) $\linf$-norm over $\Ker D_J$ is likely to produce lower values of $\IC$. In a nutshell, linear dependencies among the columns of $D$, in some sense, are desirable to optimize the value of $\IC$. This is in agreement with the observations of \cite{nam2011the-c}.
\end{rem}

At this stage, one may wonder whether the sufficient condition $\IC(\signxx) < 1$ can be weakened while ensuring sign consistency and cospace recovery by solving \lasso in presence of small noise. The following proposition provides a first answer by proving that the condition is in some sense necessary.
\begin{prop}\label{prop:icsupnoise}
  Let $x_0 \in \RR^N$ be a fixed vector of \dcosup $J$. Let $y=\Phi x_0+w$.
  Suppose that $(H_J)$ holds and $\IC(\signxx) > 1$.
  If
  \begin{equation*}
  	\frac{1}{\lambda} \normi{\BJ^{[J]} w} < \IC(\signxx) - 1,
  \end{equation*}
  then for any solution $\xsoly$ of \lasso, we have
  \begin{equation*}
   	\signxx \neq \sign(D^* \xsoly) .
   \end{equation*} 
\end{prop}
In plain words, for signals $x_0$ with $\IC(\signxx) > 1$, the associated sign vector and \dsup cannot be identified by solving \lasso even with a small noise.


\subsection{Noiseless Identifiability} 
\label{sub:contrib-noiseless}

In the noiseless case, $w=0$, the criterion $\IC$ can be used to test identifiability.
A vector $x_0$ is said to be \emph{identifiable} if $x_0$ is the unique solution of (\lassoP{\Phi x_0}{0}).
We will prove the following theorem.
\begin{thm}\label{thm:noiseless}
  Let $x_0 \in \RR^N$ be a fixed vector of \dcosup $J$.
  Suppose that \eqref{eq:hj} holds and $\IC(\signxx) < 1$.
  Then $x_0$ is identifiable.
\end{thm}

The conclusions of Proposition~\ref{prop:icsupnoise} remain valid even in the noiseless case.
\begin{cor}\label{cor:icsup}
  Let $x_0 \in \RR^N$ be a fixed vector of \dcosup $J$.
  Suppose that $(H_J)$ holds and $\IC(\signxx) > 1$.
  Then for any $\lambda > 0$ and any solution $\xsoly$ of $(\lassoP{\Phi x_0}{\lambda})$,
  \begin{equation*}
   	\signxx \neq \sign(D^* \xsoly) .
   \end{equation*}
\end{cor}
When $\IC(\signxx) = 1$, Proposition~\ref{prop:icsupnoise} and Corollary~\ref{cor:icsup} do not allow to conclude. In Section~\ref{sub:tv}, a family of signals $x_0$ is built such that $\IC(\signxx) = 1$, and where we show that depending on the noise structure, recovery can be possible or not.


\subsection{Robustness to Bounded Noise} 
\label{sub:contrib-noise}

Let us now turn to robustness to an arbitrary bounded noise. To this end, we introduce the following criterion which is a strenghthned version of the $\IC$ criterion.
\begin{defn}\label{def:arc}
  The analysis \emph{Recovery Criterion (RC)} of $I \subset \ens{1,\dots,P}$ is defined as
  \begin{equation*}
    \ARC(I) = \umax{\substack{p_I \in \RR^{\abs{I}}\\ \normi{p_I} \leq 1}} \umin{u \in \Ker D_J} \normi{\CJ^{[J]} p_I - u} .
  \end{equation*}
\end{defn}
It is clear that if $I$ is the \dsup of $x_0$, $\ARC(I) < 1$ implies $\IC(\signxx) < 1$. Moreover, $\ARC$ depends solely on the \dsup while $\IC$ relies both on the \dsup and the sign vector $\signxx$.

In Theorem~\ref{thm:small}, the assumption on $T$ plays a pivotal role: if $T$ is too small, there is no way to distinguish the small components of $D^* x_0$ from the noise. If no assumption is made on $T$, it turns out that one can nevertheless expect robustness to an arbitrary bounded noise if the parameter $\lambda$ is large enough. In this case, solving \lasso allows to recover a unique vector which lives in the same $\GJ$ as the unknown signal $x_0$, and whose $\ldeux$ distance from $x_0$ is within a factor of the noise level.
\begin{thm}\label{thm:noise}
  Let $I$ be a fixed \dsup, $J = I^c$ its associated \dcosup. Let $y=\Phi x_0+w$.
  Suppose that \eqref{eq:hj} holds.
  If $\ARC(I) < 1$ and
  \begin{equation*}
    \lambda = \rho \norm{w}_2 \frac{c_J}{1 - \ARC(I)} \qwithq \rho > 1,
  \end{equation*}
  where
  \begin{equation*}
    c_J = \norm{D_J^+ \Phi^* (\Phi \AJ \Phi^* - \Id)}_{2, \infty} ,
  \end{equation*}
  then for every $x_0$ of \dsup $I$, problem \lasso has a unique solution $\xsoly$ whose \dsup is included in $I$ and $\norm{x_0 - \xsoly}_2 = O(\norm{w}_2)$.
  More precisely,
  \begin{equation*}
    \norm{x_0 - \xsoly}_2 \leq \norm{\AJ}_{2, 2} \norm{w}_2 \left( \norm{\Phi}_{2, 2} + \dfrac{\rho c_J}{1 - \ARC(I)} \norm{D_I}_{2, \infty} \right) \!.
  \end{equation*}
\end{thm}



\section{Related Works} 
\label{sec:related}

\subsection{Previous Works on Synthesis Identifiability and Robustness} 
\label{sub:rw-tropp}

There is an extensive literature on guarantees for identifiability and robustness to noise of sparse synthesis regularization, i.e. Lasso in \eqref{eq:lasso-s}.
In \cite{fuchs2004on-sp}, Fuchs introduced a synthesis identifiability criterion $\IC_S$ which is a specialization of our $\IC$ (see Definition \ref{def:asc}) to the case where $D=\Id$. This condition also known as the irrepresentable condition in the statistical literature.
\begin{defn}
  Let $s \in \{-1, 0, +1\}^P$, $I$ its support and $J$ its cosupport.
  We suppose $\Psi_I$ is full rank.
  The criterion $\IC_S$ of a sign vector $s$ associated to a support $I$ is defined as
  \begin{equation*}
    \IC_S(s) = \normi{\Omega^S s_I} \qwhereq \Omega^S = \Psi_J^* \Psi_I^{+,*} .
  \end{equation*}
\end{defn}
Let us point out that the full rank assumption on $\Psi_I$ is a particularization of \eqref{eq:hj} to the synthesis prior case.

The following result is proved in \cite{fuchs2004on-sp}. We restate it here for completeness.
\begin{thm-nonum}[\cite{fuchs2004on-sp}]
  Let $\alpha_0 \in \RR^P$ be a fixed vector of support $I$.
  If $\Psi_I$ has full rank and $\IC_S(\sign(\alpha_0)) < 1$, then $\alpha_0$ is identifiable, i.e. it is the unique solution of \eqref{eq:lasso-s} for $y = \Psi \alpha_0$.
\end{thm-nonum}

The work of Tropp \cite{tropp2006just-,tropp2004greed} in the synthesis case developed a sufficient noise robustness condition built upon the so-called Exact Recovery Coefficient (ERC) of the support.
\begin{defn}
  The \emph{Exact Recovery Coefficient} (ERC) of $I \subset \ens{1 \dots P}$ is defined as
  \begin{equation*}
    \mathbf{ERC}(I) = \norm{\Omega^S}_{\infty, \infty} ,
  \end{equation*}
\end{defn}
Note again that while $\IC_S(s)$ depends both on the sign and the support, $\mathbf{ERC}$ depends only on the support and we have the inequality $\IC_S(s) \leq \mathbf{ERC}(I)$.

It is proved in \cite{tropp2006just-} that $\mathbf{ERC}(I) < 1$ is a sufficient condition for partial support recovery and $\ldeux$-consistency by solving the Lasso.
\begin{thm-nonum}[\cite{tropp2006just-}]
  Let $I$ be a fixed support.
  Suppose that $\Psi_I$ has full rank.
  If $\mathbf{ERC}(I) < 1$ and $\lambda$ large enough, then for every $\alpha_0$ of support $I$, problem \eqref{eq:lasso-s} with $y = \Psi \alpha_0 + w$ has a unique solution $\alpha^\star$ whose support is included in $I$ and $\norm{\alpha_0 - \alpha^\star}_2 = O(\norm{w}_2)$.
\end{thm-nonum}

By nociting that when $D = \Id$, $\Ker D_J =\{0\}$, and by definition of the operator norm $\norm{\cdot}_{\infty, \infty}$, we easily conclude that our criteria $\IC$ and $\ARC$ are equivalent to $\IC_S$ and $\mathbf{ERC}$.
\begin{prop}\label{prop:eq-transf}
  If $D=\Id$, then $\IC(\signxx) = \IC_S(\signxx)$ and $\ARC(I) = \mathbf{ERC}(I)$.
\end{prop}

There are of course many other sufficient conditions in the literature which provably guarantee uniqueness, identifiability and noise robustness in the $\lun$-synthesis regularization case; see \cite{bruckstein2009sparse}~for a thorough review. Among the most popular we have coherence-based conditions and those based on the RIP which plays a central role in the compressed sensing theory \cite{candes2006robust,donoho2006compressed}.

In the inverse problems community, efforts have been undertaken to derive results of robustness to arbitrary bounded noise (so-called convergence rates), for $\lun$-synthesis regularization to solve ill-posed linear inverse problems. In the regularization theory, the source or range condition as well as a restricted invertibility condition on $\Phi$ are generally imposed, see e.g. \cite{Lorenz08,Grasmair08,Grasmair2011,grasmair2011necessary}, and \cite{Scherzer09} and references therein. For instance, the authors in \cite{grasmair2011necessary} have shown that a strengthened version of the source condition generalizing $\IC_S(s) < 1$ is a necessary and sufficient condition for noise robustness with the rate $O(\norm{w}_2)$. This source condition is detailed in \eqref{eq-source-cond} for the more general analysis setting. However, these results do not say anything about the sign and support recovery.



\subsection{Previous Works on Analysis Identifiability and Robustness} 
\label{sub:rw-nam}

It is only very recently that recovery and noise robustness theoretical guarantees of $\lun$-analysis sparse regularization have been investigated. The previous works that we are aware of are \cite{candes2010compr,nam2011the-c,grasmair2011linear,burger2004convergence,giryes2012iterative,peleg2012cosparse,NeedelWard12}. 

Taking a compressed sensing perspective with a generalization of the RIP (called D-RIP) on $\Phi$, and assuming that $D$ is a tight frame, the authors \cite{candes2010compr} prove that $\lun$-analysis regularization allow accurate and robust recovery from noisy measurements uniformly over {\em all} signals that are (even nearly) $D$-sparse. \cite{NeedelWard12} also give a provable guarantee of robust recovery for images from compressed measurements via total variation regularization. As usual, the RIP-based guarantees are uniform and the (D-)RIP is satisfied for Gaussian matrices and other random ensembles. This setting is thus quite far from ours.

The work of \cite{nam2011the-c} is much closer to ours. It studies noiseless identifiability using $\lzero$ and $\lun$ sparse analysis regularization.
Their result on $\lun$-analysis noiseless identifiability is the following whose proof is inspired from a generalization of the null space property to the $\lun$-analysis case.
\begin{thm-nonum}[\cite{nam2011the-c}]
  Let $M^*$ be a basis matrix of $\Ker \Phi$ and $I$ a fixed \dsup such that the matrix $D_J^* M^*$ has full rank.
  Let $x_0 \in \GJ$ be a fixed vector.
  If $\FNam(\sign(D^* x_0)) < 1$ and
  \begin{equation*}
    \FNam(s) = \normi{\beth_I s_I} \qwhereq \beth_I = (M D_J)^+ M D_I ,
  \end{equation*}
  then $x_0$ is identifiable.
\end{thm-nonum}
Note that $\FNam(s) < 1$ does not imply $\IC(s) < 1$ and neither the opposite. Moreover, unlike $\IC$, $\FNam$ does not reduce to $\IC_S$ in the $\lun$-synthesis case, see the discussion on their fundamental differences in \cite[Section~5.3]{nam2011the-c}.
In the noisy case, $\FNam(s)<1$ is not sufficient to guarantee stability of the \dsup and the sign vector to noise even to a small one.
More precisely, let $x_0$ be a fixed vector, and denote $s = \sign(D^* x_0)$ where $I$ is its \dsup and $y = \Phi x_0 + w$.
If $\FNam(s) < 1$ but $\IC(s) > 1$, then according to Proposition~\ref{prop:icsupnoise}, any solution $\xsoly$ of \lasso, for $\lambda$ close to zero, is such that the \dsup of \xsol is not included in $I$. Robustness guarantees of $\lun$-analysis regularization by an appropriate strengthening of the analysis equivalent of the null space property remains an open porblem. 

Turning to the inverse problems literature, some authors have established linear convergence rates. For instance, in \cite{burger2004convergence}, convergence (robustness) rates for convex regularizations $R$ have been derived with respect to the Bregman divergence under a source condition. The Bregman divergence measures the distance between the regularization term $\J$ and its affine approximation at the true solution.
Analysis-type regularizations where $D^*$ is not necessarily injective, such as the total variation, fall within the class of regularization functionals they considered. The author in \cite{grasmair2011linear} derived more general linear convergence rates for a large class of convex sparsity promoting regularization functionals $R$, including analysis-type ones, under a source condition and a suitable restricted injectivity condition on $\Phi$. The convergence was established with respect to the error in the solution measured in terms of the regularization functional. 
Specialized to the case of $\lun$-analysis regularization, this result reads.
\begin{thm-nonum}[\cite{grasmair2011linear}]
  Let $x_0 \in \RR^N$ of \dsup $I$ and $y=\Phi x_0 + w$ such that $\norm{w} = \epsilon$.
  Assume also that there exist $\alpha$ such that
  \eql{\label{eq-source-cond}
  	\alpha \in \partial \normu{\cdot}(D^* x_0)
	\qandq 
	D \alpha \in \Im \Phi^*
	}
	(source condition).
  Let $J \subseteq I^c$ such that $\normi{\alpha_J} < 1$. Suppose that \eqref{eq:hj} holds with such $J$.
  Then, for $\lambda$ proportional to $\epsilon$, there exists $C$ independent of $\epsilon$ such that
  \begin{equation*}
    \norm{D^*(\xsoly - x_0)}_2 \leq C \epsilon ~.
  \end{equation*}
\end{thm-nonum}
Interestingly, for $J=I^c$, if \eqref{eq:hj} is satisfied, $\IC(\sign(D^* x_0)) < 1$ implies that the source and restricted injectivity conditions stated in the previous theorem are in force. More precisely, the following holds.
\begin{prop}\label{prop:source}
  Let $x_0 \in \RR^N$ of \dcosup $J$ such that \eqref{eq:hj} holds and $\IC(\sign(D^* x_0)) < 1$.
  Then, the source and restricted injectivity conditions of Theorem ([35]) hold. The claimed convergence is therefore also valid.
\end{prop}
However, in none of these works in the inverse problem literature, robustness with respect to the $\ldeux$-norm, i.e. $\ldeux$-distance of the solution from the true one, was established for general $D$.  Of course, if $D^*$ were injective, $\ldeux$-robustness would follow immediately from \cite{grasmair2011linear}.
In addition, their results do not allow to conclude anything about the sign and \dsup recovery unless there is no noise.



\section{Examples} 
\label{sec:examples}

This section details algorithms to compute the criteria $\IC$ and $\ARC$, together with a detailed study of three $\lun$-analysis regularizations: total variation, that when $D$ is the shift-invariant Haar dictionary, and the Fused Lasso.
The source code used to produce the numerical results is available online at \url{github.com/svaiter/robust_sparse_analysis_regularization}.

\subsection{Computing Sparse Analysis Regularization} 
\label{sub:schemes}

It is not the main scope of this paper to give a comprehensive treatment of provably convergent minimization schemes that can be used to solve \lasso. We describe one possible efficient algorithm to do so which originates from the realm of nonsmooth convex optimization theory, and more precisely, proximal splitting.

In the case where $\Phi = \Id$ (denoising), \lasso is strictly (actually strongly) convex, and one can compute its unique solution $\xsoly$ by solving an equivalent Fenchel-Rockafellar dual problem \cite{chambolle2004algorithm}
\begin{equation*}
  \xsoly = y + D \alpha^\star 
  \qwhereq 
  \alpha^\star \in \uargmin{\norm{\alpha}_\infty \leq \lambda} \norm{y + D\alpha}_2^2 .
\end{equation*}
The dual problem can be solved using e.g. projected gradient descent or a multi-step accelerated version of it.

In the general case, we advocate the use of a primal-dual algorithm such as the relaxed Arrow-Hurwicz scheme recently revitalized in \cite{chambolle2011first}. This algorithm is designed to minimize the sum of two proper lowersemicontinuous convex functions, one of which is comped by a linear bounded operator. To put problem \lasso in a form amenable to apply this scheme, we can rewrite it as follows
\begin{equation*}
  \umin{x \in \RR^N} F(K(x))
  \qwhereq
  \choice{
    F: (g,u) \mapsto \frac{1}{2}\norm{y-g}_2^2 + \lambda \normu{u} \\
    K: x \mapsto ( \Phi x, D^* x ).
  }
\end{equation*}
The primal-dual algorithm requires the computation of the proximity operator of $F$ which is a separable and simple function, i.e. its proximity operator is easy to compute. Recall that the proximity operator $\Prox_{f}$ of a proper lower semicontinuous function and convex $f$ is defined as 
\begin{equation*}
  \forall x \in \RR^N, \quad \Prox_{f}(x) = \uargmin{z \in \RR^N} \dfrac{1}{2} \norm{z-x}_2^2 + f(z) .
\end{equation*}
Computing $\Prox_F$ involves applying a soft-thresholding (the $\lun$-part) and a diagonal Wiener filtering (the separable quadratic part).


\subsection{Computing the Criteria} 
\label{sub:crit-num}

In the case where $\Ker D_J \neq \ens{0}$, computing $\IC(\signxx)$ entails solving a convex minimization problem.
The latter can be cast as
\begin{align*}
  \IC(\signxx) = & \umin{u \in \RR^N} \norm{\CJ^{[J]} \signxx_I - u}_\infty \\
  & + \iota_{\Ker D_J}(u) ,
\end{align*}
where $\iota_{\Ker D_J}$ is the indicator function of $\Ker D_J$, i.e.
\begin{equation*}
  \iota_{\Ker D_J}(u) =
  \begin{cases}
    0 & \text{if } u \in \Ker D_J \\
    + \infty & \text{otherwise}.
  \end{cases}
\end{equation*}
The objective above is the sum of a translated $\linf$-norm and the indicator function of $\Ker D_J$. It can then be solved efficiently with the Douglas-Rachford splitting algorithm \cite{combettes2007douglas}. This will necessitate to compute the proximity operator of $\iota_{\Ker D_J}$ which is the orthogonal projector on $\Ker D_J$, and $\Prox_{\normi{\CJ^{[J]} \signxx_I - \cdot}}$ can be computed with standard proximal calculus rules knowing that
\begin{equation*}
  \Prox_{\gamma \normi{\cdot}}(x) = x - P_{\normu{\cdot}}\pa{\frac{x}{\gamma}} , \quad \forall \gamma > 0
\end{equation*}
where $P_{\normu{\cdot}}$ is the projection onto the unit $\lun$ ball.
This projector can be computed through sorting and soft-thresholding, see \cite{fadili2010tv} for details.\\

Unfortunately, computing \ARC~(see Definition~\ref{def:arc}) is not as easy since it necessitates to solve a difficult maxi-minimization optimization problem which is nonsmooth, and convex in both $u$ and $p_I$ (while concavity in $p_I$ would have been desirable).
A stronger criterion, which is easy to compute, is obtained by taking $u=0$ in $\Ker D_J$
\begin{equation*}
  \text{w\ARC}(I) = \norm{\CJ^{[J]}}_{\infty, \infty}.
\end{equation*}
One can easily see that for every vector $x_0$ with \dsup $I = \supp(D^* x)$, the following inequalities hold
\begin{equation*}
  \IC(\signxx) \leq \ARC(I) \leq \text{w\ARC}(I).
\end{equation*}
For many cases, $\text{w\ARC}(I)$ might be strictly greater than 1.
However, there are situations where $\text{w\ARC}(I) < 1$, such as when the associated cospace $\GJ$ is close to the whole space, i.e. high $D$-cosparsity or equivalently very small $D$-sparsity.


\subsection{Total Variation Denoising} 
\label{sub:tv}

Discrete 1-D total variation (TV) corresponds to taking $D = D_\text{DIF}$ as defined in \eqref{eq-d-diff}.
We recall that the TV union of subspaces model is formed by $\bigcup_k \Theta_k$ where $\Theta_k$ is the subspace of piecewise constant signals with $k-1$ steps.
We now define a subclass of such signals.
\begin{defn}
  A signal is said to \emph{contain a staircase subsignal} if there exists $i \in \ens{1\dots \abs{I}-1}$ such that
  \begin{equation*}
    \sign(D_I^* x)_i = \sign(D_I^* x)_{i+1} = \pm 1 .
  \end{equation*}
\end{defn}
Figure \ref{fig:num-tv} shows examples of signals with and without staircase subsignals.
\begin{figure}[bth]
  \centering
  \includegraphics[width=\linewidth]{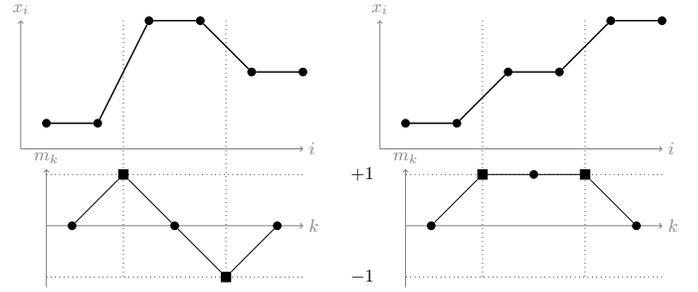}
  \caption{Top row: Two examples of signals $x$ having 2 jumps. Bottom row: Associated dual vector $m$.}\label{fig:num-tv}
\end{figure}

The following result will allow to characterize robustness of TV regularization when $\Phi=\Id$, i.e. TV denoising.

\begin{prop}
  We consider the case where $\Phi = \Id$.
  If $x_0$ does not contain a staircase subsignal, then $\IC(\sign(D^* x_0)) < 1$.
  Otherwise, $\IC(\sign(D^* x_0)) = 1$.
\end{prop}
\begin{proof}
  Let $\xsoly$ be the unique solution of \lasso with \dcosup $J$ and $I = J^c$.
  Using Lemma \ref{lem:first-order}, there exists $\sigma \in \Sig(\xsoly) \subset \RR^{|J|}$. 
  Since $D_J^+ \AJ = 0$, we have $\CJ^{[J]} = -D_J^+ D_I$.
  We denote the vector $m$ defined as
  \begin{equation*}
    m :
    \choice{
      m_I = s_I = \sign(D_I^* x) \\
      m_J = \sigma = \CJ^{[J]} s_I .
    }
  \end{equation*}
  The vector $\sigma$ satisfies $(D_J^* D_J) \si = -(D_J^* D_I) s_I$.
  One can show that this implies that $m$ is the solution of a discrete Poisson equation
  \begin{equation*}
    \foralls j \in J, \quad (\Delta m)_j = 0
    \qandq
    \choice{
    \foralls i \in I, \: m_i = s_i, \\
    m_0 = m_{N} = 0.
    }
  \end{equation*}
  where $\Delta = D D^*$ is a discrete Laplacian operator.
  This implies that for $i_1<k<i_2$ where $i_1,i_2$ are consecutive indices of $I$, $m$ is obtained by linearly interpolating (see Figure \ref{fig:num-tv}) the values $m_{i_1}$ and $m_{i_2}$, i.e
  \begin{equation*}
    m_k = \rho m_{i_1} + (1-\rho) m_{i_2} \qwhereq \rho = \frac{k-i_1}{i_2-i_1} .
  \end{equation*}
  Hence, if $x_0$ does not contain a staircase subsignal, one has $\normi{\CJ^{[J]} s_I} < 1$.
  On the contrary, if there is $i_1$ such that $s_{i_1} = s_{i_2}$, where $i_1$ and $i_2$ are consecutive indices of $I$, then for every $i_1 < j < i_2, m_j = s_{i_1} = \pm 1$ which implies that $\IC(\sign(D^* x_0)) = 1$.
\end{proof}
This proposition together with Theorem~\ref{thm:small} shows that if a signal $x_0$ does not have a staircase subsignal, then TV denoising from $y=x_0+w$ is robust to a small noise. This means that if $w$ is small enough, for $\la$ proportional to the noise level, the TV denoised version of $y$ contains the same jumps as $x_0$. However, the presence of a staircase in a signal, i.e. $\IC(\sign(D^* x_0)) = 1$, does not comply with the assumptions of neither Theorem~\ref{thm:small} nor Proposition \ref{prop:icsupnoise}. This prevents us from drawing positive or negative robustness conclusions.\\

To gain a better understanding of the latter situation, we build an instructive family of signals for which the $\IC$ criterion saturates at 1. It will turn out that depending on the structure of the noise $w$, the \dsup of $x_0$ can be either stably identified or not.\\
For $N$ a multiple of 4, we split $\ens{1,\dots,N}$ into 4 sets
$l_k = \ens{(k-1)M+1, ..., kM}$ of cardinality $M = N/4$. Let $\mathbf{1}_{l_k}$ be the boxcar signal whose support is $l_k$ .
Consider the staircase signal $x_0 = - \mathbf{1}_{l_1} + \mathbf{1}_{l_4}$ degraded by a deterministic noise $w$ of the form $w = \epsilon (\mathbf{1}_{l_3} - \mathbf{1}_{l_2})$, where $\epsilon \in \RR$.
The observation vector $y = x_0 + w$ reads
\begin{equation*}
  y = - \mathbf{1}_{l_1} - \epsilon \mathbf{1}_{l_2} + \epsilon \mathbf{1}_{l_3} + \mathbf{1}_{l_4} ~.
\end{equation*}
Suppose that $\epsilon > 0$, then the solution $x_\lambda^\star$ of \lasso is
\begin{equation*}
  x_\lambda^\star = 
  \left(-1 + \frac{\lambda}{M} \right) \mathbf{1}_{l_1} 
  - \epsilon \mathbf{1}_{l_2} 
  + \epsilon \mathbf{1}_{l_3} 
  + \left(1 - \frac{\lambda}{M}\right) \mathbf{1}_{l_4} , 
\end{equation*}
if $0 \leq \lambda \leq \lambda_1 = M(1 - \epsilon)$, and
\begin{equation*}
  x_\lambda^\star = 
  \left(
  -\epsilon + \frac{\lambda - \lambda_1}{2M}
  \right) (\mathbf{1}_{l_1} + \mathbf{1}_{l_2})
  +
  \left(
  \epsilon - \frac{\lambda - \lambda_1}{2M}
  \right) (\mathbf{1}_{l_3} + \mathbf{1}_{l_4}) ,
\end{equation*}
if $\lambda_1 \leq \lambda \leq \lambda_2 = \lambda_1 + 2 \epsilon M$, and 0 if $\lambda > \lambda_2$.
Similarly, if $\epsilon < 0$, the solution $x_\lambda^\star$ reads
\begin{equation*}
  x_\lambda^\star = 
  \left(-1 + \frac{\lambda}{M} \right) \mathbf{1}_{l_1} 
  - \left( \epsilon + 2\frac{\lambda}{M}\right) (\mathbf{1}_{l_2}  - \mathbf{1}_{l_3} )
  + \left(1 - \frac{\lambda}{M}\right) \mathbf{1}_{l_4} , 
\end{equation*}
if $0 \leq \lambda \leq \bar\lambda_1 = - \epsilon \frac{M}{2}$, and
\begin{equation*}
  x_\lambda^\star = 
  \left(-1 + \frac{\lambda}{M} \right) \mathbf{1}_{l_1} 
  + \left(1 - \frac{\lambda}{M}\right) \mathbf{1}_{l_4} , 
\end{equation*}
if $\lambda_1 \leq \lambda \leq \bar\lambda_2 = M$, and 0 if $\lambda > \lambda_2$.
Figure \ref{fig:tv-path} displays plots of the the coordinates' path for both cases.
\begin{figure}[bth]
  \centering
  \includegraphics[width=\linewidth]{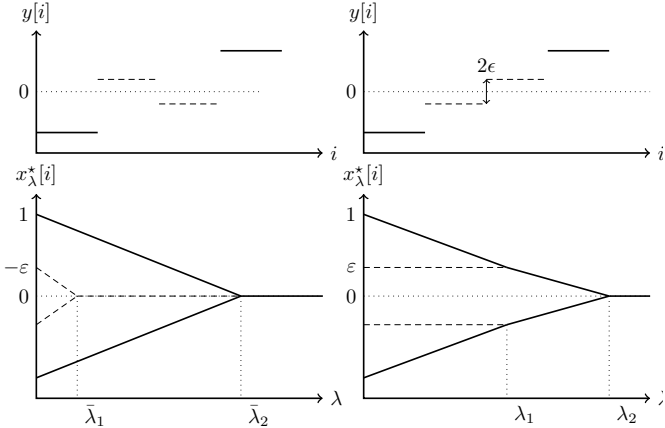}
  \caption{
  Top row: Signals $y$ for $\epsilon < 0$ (left) and $\epsilon > 0$ (right).
  Bottom row: Corresponding coordinates' path of $x_\lambda^\star$ as a function of $\lambda$.
  The solid lines correspond to the coordinates in $l_1$ and $l_4$, and the dashed ones to the coordinates in $l_2$ and $l_3$.
  }\label{fig:tv-path}
\end{figure}
It is worth poiniting out that when $\epsilon > 0$, the \dsup of $x_\lambda^\star$ is always different from that of $x_0$ whatever the choice of $\la$, whereas in the case $\epsilon < 0$, for any $\bar \lambda_1 \leq \lambda \leq \bar \lambda_2$, the \dsup of $x_\lambda^\star$ and sign of $D^* x_\lambda^\star$ are exactly those of $x_0$.


\subsection{Shift-Invariant Haar Deconvolution} 
\label{sub:comp}

Sparse analysis regularization using a 1-D shift invariant Haar dictionary is efficient to recover piecewise constant signals.
This dictionary is built using a set of scaled and dilated Haar filters
\begin{equation*}
  \psi_i^{(j)} = \frac{1}{2^{\tau(j+1)}}
  \begin{cases}
    +1 & \text{if } 0 \leq i < 2^j \\
    -1 & \text{if } -2^j \leq i < 0 \\
    0  & \text{otherwise} ,
  \end{cases}
\end{equation*}
where $\tau \geq 0$ is a normalization exponent. For $\tau = 1$, the dictionary is said to be \emph{unit-normed}.
For $\tau = 1/2$, it corresponds to a \emph{Parseval tight-frame}.
The action on a signal $x$ of the analysis operator corresponding to the translation invariant Haar dictionary $D_{H}$ is 
\begin{equation*}
  D_{H}^* x = \pa{\psi^{(j)} \star x}_{0 \leq j \leq J_{\text{max}}} ,
\end{equation*}
where $\star$ stands for the discrete convolution (with appropriate boundary conditions) and $J_{\text{max}} < \log_2(N)$.
The analysis regularization $\normu{D_{H}^* x}$ can also be written as the sum over scales of the TV semi-norms of filtered versions of the signal.
As such, it can be understood as a sort of multiscale total variation regularization.
Apart from a multiplicative factor, one recovers Total Variation when $J_{\text{max}} = 0$.

We consider a noiseless convolution setting where $\Phi$ is a circular convolution operator with a Gaussian kernel of standard deviation $\si$.
We first study the impact of $\sigma$ on the identifiability criterion $\IC$.
The blurred signal $x_\eta$ is a centered boxcar signal with a support of size $2\eta N$
\begin{equation*}
   x_{\eta} = \mathbf{1}_{\{\lfloor N/2-\eta N\rfloor,\dots,\lfloor N/2+\eta N\rfloor\}}, \quad \eta \in (0,1/2] ~.
\end{equation*}
Figure~\ref{fig:num-haar-deconvol} displays the evolution of $\IC(\sign(D_{H}^* x_0)$ as a function of $\sigma$, where we fixed $\eta = 0.2$.
\begin{figure}[!bth]
  \centering
  \includegraphics[width=\linewidth]{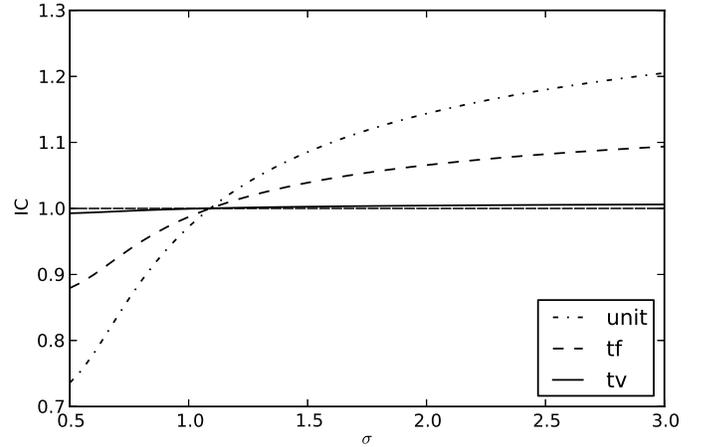}
  \caption{Behaviour of $\IC$ for a noiseless deconvolution scenario with a Gaussian blur and $\ell_1$-analsyis sparsity regularization in a shift invariant Haar dictionary with $J_{\max} = 4$. $\IC$ is plotted as a function of the Gaussian bluring kernel size $\sigma \in  [0.5, 3.0]$ for the total variation dictionary and the Haar wavelet dictionary with two normalization exponents $\tau$. Dash-dotted line: $\tau=1$ (unit-normed). Dashed line: $\tau=1/2$ (tight-frame). Solid line: total variation.}\label{fig:num-haar-deconvol}
\end{figure}
In the identifiability regime, $\IC(\sign(D_{H}^* x_0)$ appears smaller in the case of the unit-normed normalization. However, one should avoid to infer stronger conclusions since a detailed computation of the constants involved in Theorem~\ref{thm:small} would be necessary to completely and fairly compare the stability performance achieved with each of these three dictionaries.


\subsection{Fused Lasso Compressed Sensing} 
\label{sub:fused}

Fused Lasso was introduced in \cite{tibshirani2005sparsity}.
It corresponds to taking
\begin{equation*}
  D = 
  \begin{bmatrix}
    D_{\text{DIF}} & \epsilon \Id
  \end{bmatrix} ,
\end{equation*}
in \lasso, where $\epsilon > 0$.
The associated union of subspaces \eqref{eq:subspaces} is $\bigcup_k \Theta_k$, where $\Theta_k$ is the set of signals that are the sum of $k$ boxcars of disjoint supports, i.e
a signal $x \in \Theta_k$ can be written as
\begin{equation*}
  x = \sum_{i=1}^k \gamma_i \mathbf{1}_{[a_i, b_i]} ,
\end{equation*}
where $\gamma_i \in \RR$ and $a_i \leq b_i < a_{i+1}$.

We consider a noiseless compressed sensing setting and examine the behaviour of $\IC$ with respect to the sampling ratio $Q/N$ and the true signal properties. $\Phi$ is drawn from the standard Gaussian ensemble, i.e. $\Phi_{i,j} \sim_{i.i.d.} \Nn(0,1)$. The sampled signal $x_{\eta,\rho}$ is the superposition of two boxcars distant from each other by $2\rho N$ and each of support size $\eta N$
\begin{equation*}
  x_{\eta,\rho} = \mathbf{1}_{\{\lfloor(\frac{1}{2} - \eta - \rho)N\rfloor,\dots,\lfloor(\frac{1}{2} - \rho)N\rfloor\}} + \mathbf{1}_{\{\lfloor(\frac{1}{2} + \rho)N\rfloor,\dots,\lfloor(\frac{1}{2} + \eta + \rho)N\rfloor\}}.
\end{equation*}
In our simulations, we fixed $\rho = 0.1$.

Figure~\ref{fig:num-fused-transition} depicts the evolution of the empirical probability with respect to the sampling of $\Phi$ of the event $\IC < 1$ as a function of the sampling ratio $Q/N \in [0.5, 1]$ and the support size $\eta \in [0.025, 0.15]$.
This probability is computed from 1000 Monte-Carlo replications of the sampling of $\Phi$. With no surprise, one can clearly see that the probability increases as more measurements are collected. This probability profile also seems to be increasing as $\eta$ decreases, but this is likely to be a consequence of the choice of the Fused Lasso parameter $\epsilon$, and the conclusion may be different for other choices. 
\begin{figure}[!bth]
  \centering
  \includegraphics[width=\linewidth]{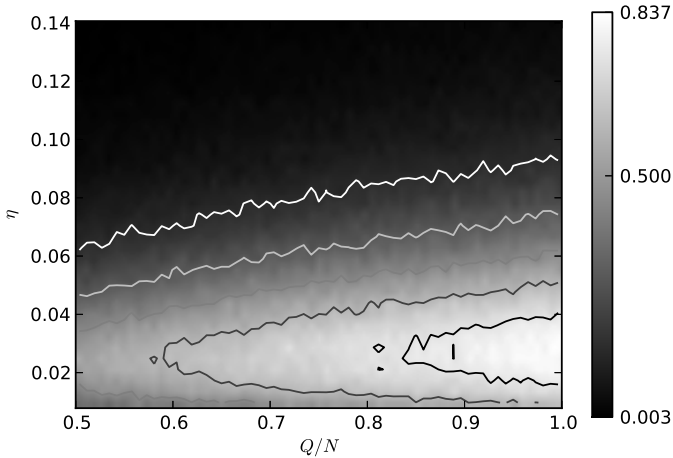}
  \caption{Behaviour of $\IC$ for a compressed sensing scenario matrix with a Gaussian measurement matrix and the Fused Lasso regularization. Empirical probability of the event $\IC < 1$ as a function of the sampling ratio $Q/N \in [0.5, 1]$ and the support size $\eta \in [0.025, 0.15]$.}\label{fig:num-fused-transition}
\end{figure}

This is indeed confirmed in our last experiment whose results are displayed in Figure~\ref{fig:num-fused-transition-eps}. It shows the evolution of the empirical probability of the event $\IC < 1$ as a function of the Fused Lasso parameter $\epsilon \in [1/N, 200/N]$ and the support size $\eta \in [0.025, 0.15]$. This probability is again computed from 1000 Monte-Carlo replications. Depending on the choice of $\epsilon$, the probability profile is not necessarily monotonic as a function of $\eta$. For large values (more weight on $\Id$ in the Fused Lasso dictionary), the probability decreases monotonically as $\eta$ increases which can be explained by the fact that higher $\eta$ corresponds to less sparse signals. As $\epsilon$ is lowered, higher weight is put on the TV regularization, and the behaviour is not anymore monotonic. Now, the probability reaches a peak at intermediate values of $\eta$ and then vanishes quickly. The peak probability also decreases with decreasing $\epsilon$. 
 
\begin{figure}[!bth]
  \centering
  \includegraphics[width=\linewidth]{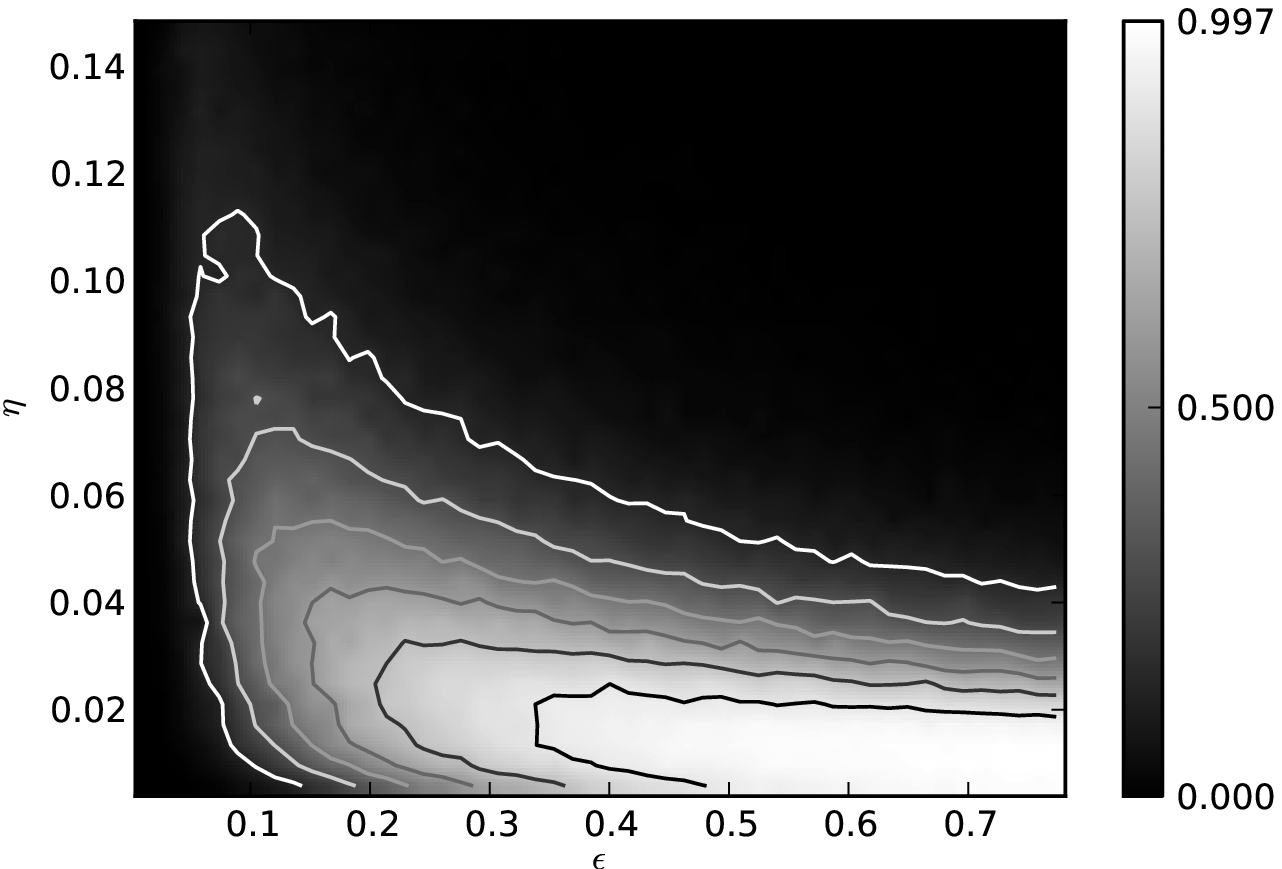}
  \caption{Behaviour of $\IC$ for a compressed sensing scenario matrix with a Gaussian measurement matrix and the Fused Lasso regularization. Empirical probability of the event $\IC < 1$ as a function of the parameter $\epsilon \in [1/N, 200/N]$ and the support size $\eta \in [0.025, 0.15]$.}\label{fig:num-fused-transition-eps}
\end{figure}



\section{Proofs} 
\label{sec:proofs}

This section details the proofs of our main results in Theorems~\ref{thm:small}-\ref{thm:noise}.
Throughout, we use the shorthand notation $\obj$ for the objective function in \lasso
\begin{equation*}
  \obj(x) = \minform .
\end{equation*}
We remind the reader that condition \eqref{eq:H0} is supposed to hold true in all our statements.

\subsection{Preparatory lemmata}
\label{subsec:preplem}
We first need some key lemmata that will be central in our proofs.

The first one gives the first order optimality conditions for the analysis variational problem \lasso.
\begin{lem}\label{lem:first-order}
  A vector $\xsoly$ is a solution of \lasso if, and only if, there exists $\sigma \in \RR^{\abs{J}}$, where $J$ is the \dcosup of $\xsoly$, such that
  \begin{equation}\label{eq:first-order}
    \sigma \in \Sig(\xsoly)
  \end{equation}
  \begin{multline}\label{eq:first-order-explicit}
    \Sig(x^\star) = \Big\{ \sigma \in \RR^{\abs{J}} \setminus  \Phi^*(\Phi x^\star - y) \\
     \quad + \lambda D_I s_I + \lambda D_J \sigma = 0  \\
                   \qandq \normi{\sigma} \leq 1 \Big\} 
  \end{multline}
  where $I = J^c$ is the \dsup of  $\xsoly$ and $s = \sign(D^* x^\star)$.
\end{lem}
\begin{proof}
   The subdifferential of a real-valued proper convex function $F : \RR^N \to \RR \cup \{\infty\}$ is denoted $\partial F$. From standard convex analysis, we recall the definition of the subdifferential of $F$ at a point $x$ in the domain of $F$ 
  \begin{equation*}
    \partial F(x) = \enscond{ g \in \RR^N \!}{\!\forall z \in \RR^N\!, F(z) \! \geq \! F(x) \!+\! \dotp{g}{z - x}} .
  \end{equation*}
  It is clear from this definition that $x^\star$ is a (global) minimizer of $F$ if, and only if, $0 \in \subd F(x)$.
  By classical subdifferential calculus, the subdifferential of $\obj$ at $x$ is the non-empty convex compact set
  \begin{equation*}
    \subd \obj(x) = \enscond{\Phi^* (\Phi x - y) + \lambda D u}{u \in U(x)} ,
  \end{equation*}
  where
  \begin{multline*}
    U(x) =  (\partial \norm{\cdot}_1)(D^* x)\\
    	 = \enscond{u \in \RR^N}{u_I = \sign(D^* x)_I \text{ and } \normi{u_J} \leq 1} .
  \end{multline*}
  where $I$ and $J$ are respectively the \dsup and \dcosup of $x$.
  Therefore $0 \in \subd \obj(x^\star)$ is equivalent to the existence of $u \in \RR^N$ such that $u_I = \sign(D^* x^\star)_I$ and $\normi{u_J} \leq 1$ satisfying
  \begin{equation*}
    \Phi^* (\Phi x^\star - y) + \lambda D u = 0 .
  \end{equation*}
  Letting $\sigma=u_J$, this is equivalent to $\sigma \in \Sig(x^\star)$.
\end{proof}

\medskip

The following lemma is a key to prove uniqueness statements. It characterizes the normal cone at zero to the subdifferential of $\obj$ at a minimizer $\xsoly$. By definition, this normal cone is
\begin{equation*}
    \mathcal{N}_{\partial \obj(\xsoly)}(0) = \enscond{z \in \RR^N}{\dotp{z}{d} \leq 0, \forall d \in \partial \obj(\xsoly)}.
\end{equation*}

\begin{lem}\label{lem:normal}
  Let $\xsoly$ be a solution of \lasso whose \dsup is $I^\star$.
  Suppose there exist $J \subseteq (I^\star)^c$ and $\sigma \in \Sig(\xsoly)$ with $\normi{\sigma_J} < 1$.
  Then,
  \begin{equation*}
    \mathcal{N}_{\partial \obj(\xsoly)}(0) \subseteq (\Im D_J)^\bot = \GJ.
  \end{equation*}
  Moreover, if $J$ is the \dcosup of $\xsoly$, then
  \begin{equation*}
    \mathcal{N}_{\partial \obj(\xsoly)}(0) = \GJ.
  \end{equation*}  
\end{lem}
\begin{proof}
  Let $I = J^c$.
  We decompose $I$ such that $I = I^\star \cup J^\star$ where $J^\star = (I^\star)^c \cap I$.
  Since $\normi{\sigma_J} < 1$, it follows that $\bar u$ defined by
  \begin{equation*}
    \bar u :
    \begin{cases}
      \bar{u}_{I^\star} & = \sign(D_{I^\star}^* x^\star) \\
      \bar{u}_{J^\star} & = \sigma_{J^\star} \\
      \bar{u}_J & = \sigma_J
    \end{cases}
  \end{equation*}
  is such that $\normi{\bar{u}_J} < 1$ and therefore from Lemma~\ref{lem:first-order}
  \begin{equation*}
    \Phi^* (\Phi \xsoly - y) + \lambda D \bar u = 0.
  \end{equation*}

  Let $0 < \epsilon < 1$ such that $\normi{\sigma_J} = 1 - \epsilon$.
  Consider the set 
  \begin{equation*}
    \Uu = \enscond{u \in \RR^P}{\normi{u_J - \bar{u}_J} \leq \epsilon \qandq u_I = \bar{u}_I} .
  \end{equation*}
  For every $u \in \Uu$, we define
  \begin{equation*}
    d_u = \Phi^* (\Phi \xsoly - y) + \lambda D u ,
  \end{equation*}
  and we denote
  \begin{equation*}
    \Dd = \ens{d_u}_{u \in \Uu}.
  \end{equation*}
  We therefore have
  \begin{align*}
    d_u = \lambda D (u - \bar u) 
        &= \lambda D_I (u_I - \bar{u}_I) + \lambda D_J (u_J - \bar{u}_J) \\
	&= \lambda D_J (u_J - \bar{u}_J)
  \end{align*}
  since $u_I = \bar{u}_I$.

  Let $z \in \mathcal{N}_{\partial \obj(\xsoly)}(0)$ and $u \in \Uu$.
  By construction of $u$, we have that
  \begin{equation*}
    \normi{u_J} \leq \normi{u_J - \bar{u}_J} + \normi{\bar{u}_J} \leq 1 ,
  \end{equation*}
  and
  \begin{equation*}
    u_{I^\star} = \sign(D_{I^\star}^* \xsoly)
    \qandq
    \normi{u_{J^\star}} \leq 1.
  \end{equation*}
  Clearly, $d_u \in \subd \obj(\xsoly)$.
  In view of the definition of $\mathcal{N}_{\partial \obj(\xsoly)}(0)$, we know that
  \begin{equation*}
    \dotp{z}{d} \leq 0, \quad \forall d \in \subd \obj(\xsoly) . 
  \end{equation*}
  In particular,
  \begin{equation*}
    \dotp{z}{d_u} \leq 0, \quad \forall u \in \Uu.  
  \end{equation*}
  Now, observe that $\forall u \in \Uu$, $2\bar{u}-u \in \Uu$ and $d_{2\bar{u}-u} = - d_u$.
  Indeed, $(2\bar{u}-u)_I = 2\bar{u}_I - u_I = \bar{u}_I$ and
  \begin{equation*}
    \normi{(2\bar{u}-u)_J - \bar{u}_J} 
    = \normi{\bar{u}_J-u_J} \leq \epsilon .
  \end{equation*}
  Moreover,
  \begin{align*}
    d_{2\bar{u}-u} 
      & = \Phi^* (\Phi \xsoly - y) + \lambda D (2\bar{u}-u) \\
      & = \underbrace{\Phi^* (\Phi \xsoly - y) + \lambda D \bar{u}}_{=0} - \lambda D (u - \bar u) \\
      & = - d_u .
  \end{align*}
  This implies that
  \begin{equation*}
    \forall u \in \Uu, \quad \dotp{z}{d_u} \leq 0 \qandq \dotp{z}{-d_u} = \dotp{z}{d_{2\bar{u}-u}} \leq 0.
  \end{equation*}
  That is,
  \begin{equation*}
    \forall u \in \Uu, \quad \dotp{z}{d_u} = 0 .
  \end{equation*}
  
  Let $v \in \Im D_J \setminus \ens{0}$.
  Then there exist $\mu_v \in \RR^*$ and $\sigma_v \in \RR^{\abs{J}}$ such that
  \begin{equation*}
    \mu_v v = D_J \sigma_v \qandq \normi{\sigma_v} \leq \epsilon .
  \end{equation*}
  Let the vector $u$ defined as
  \begin{equation*}
    u :
    \begin{cases}
      u_I =& \bar{u}_I \\
      u_J =& \bar{u}_J + \sigma_v.
    \end{cases}
  \end{equation*}
  $u$ is by construction an element of $\Uu$ since $\normi{u_J - \bar{u}_J} = \normi{\sigma_v} \leq \epsilon$.
  Therefore, the associated vector $d_u$ is
  \begin{equation*}
    d_u = \lambda D_J (u_J - \bar{u}_J) = \lambda D_J \sigma_v = \dfrac{\lambda}{\mu_v} v,
  \end{equation*}
  i.e. $\Im(D_J) \subseteq \Im(\Dd)$. Since $\Im(\Dd) \subseteq \Im(D_J)$, we get $\Im(D_J) = \Im(\Dd)$.
  This together with the fact that $\dotp{z}{d_u} = 0$ imply
  \begin{equation*}
    \dotp{z}{v} = \dfrac{\mu_v}{\lambda} \dotp{z}{d_u} = 0.
  \end{equation*}
  We conclude that $\mathcal{N}_{\partial \obj(\xsoly)}(0) \subseteq (\Im D_J)^\bot = \GJ$.\\

  Suppose now that $J$ is the \dcosup of $x^\star$, i.e. $J=(I^\star)^c$.
  We prove that
  \begin{equation*}
    \mathcal{N}_{\partial \obj(\xsoly)}(0) = \GJ.
  \end{equation*}
  To this end, we show that $\subd \obj(\xsoly) \subseteq \Im D_J$.
  Indeed, let $d \in \subd \obj(\xsoly)$.
  We write $d = \Phi^*(\Phi \xsoly - y) + \lambda D u$ with $u_I = \sign(D_I^* x)_I$ and $\normi{u_J} \leq 1$.
  Since $0 \in \subd \obj(\xsoly)$, one has
  \begin{equation*}
    d = \lambda D(u - \bar u) = \lambda D_J (u_J - \bar{u}_J)
  \end{equation*}
  since $u_I = \bar{u}_I$.
  This implies that $(\Im D_J)^\bot \subseteq \mathcal{N}_{\partial \obj(\xsoly)}(0)$. In view of the assertion in the first part, we conclude. 
\end{proof}

\medskip

The following lemma gives a sufficient condition which guarantees that $\lasso$ has exactly one minimizer.
\begin{lem}\label{lem:first-uniqueness}
  Let $\xsoly$ be a vector of \dsup $I^\star$.
  Suppose that there exist $J \subseteq (I^\star)^c$ such that \eqref{eq:hj} holds and $\sigma \in \RR^{\abs{(I^\star)^c}}$ such that
  \begin{equation*}
    \sigma \in \Sig(\xsoly) \qandq \normi{\sigma_J} < 1.
  \end{equation*}
  Then, $\xsoly$ is the unique solution of \lasso.
\end{lem}
\begin{proof}
  For notational convenience, we write $\obj$ as
  \begin{equation*}
    \obj(x) = q(x) + \lambda \norm{D^* x}_1 \qwhereq q(x) = \frac{1}{2} \norm{y - \Phi x}_2^2 .
  \end{equation*}
  Let $h \in \RR^N \setminus \ens{0}$.
  Two different cases occur:
  \begin{enumerate}
    \item If $h \not\in \GJ$, then using Lemma \ref{lem:normal}, $h \not\in \mathcal{N}_{\partial \obj(\xsoly)}(0)$. This negation means that $\dotp{d}{h} > 0$ for some $d \in \subd \obj(\xsoly)$, whence it follows immediately that
      \begin{equation*}
        \obj(\xsoly+h) \geq \obj(\xsoly) + \dotp{d}{h} > \obj(\xsoly) .
      \end{equation*}
    \item Let's turn to the case $h \in \GJ$.
      Since \eqref{eq:hj} holds, $q$, hence $ \obj$, is strongly convex on $\GJ$ with some modulus $c > 0$.
      Consequently, for any $v \in \subd_{\normu{D^* \cdot}}(\xsoly)$, we have
      \begin{align*}
        \obj(\xsoly \!+ h) &\geq \obj(\xsoly) +  \dotp{\nabla q(\xsoly)+\lambda v}{\!h} + \frac{c}{2}\norm{h}_2^2 \\
			   &>    \obj(\xsoly) +  \dotp{\nabla q(\xsoly)+\lambda v}{\!h} .
      \end{align*}
      $\xsoly$ is a minimizer if, and only if, $\exists v \in \subd_{\normu{D^* \cdot}}(\xsoly)$ such that
      \begin{equation*}
        \lambda v + \nabla q(\xsoly) = 0.
      \end{equation*}
      This yields,
      \begin{equation*}
        \obj(\xsoly + h) > \obj(\xsoly) .
      \end{equation*}
  \end{enumerate}
  Altogether, we have proved that for any $h \in \RR^N \setminus \ens{0}$, $\obj(\xsoly + h) > \obj(\xsoly)$, or equivalently that $\xsoly$ is the unique minimizer of \lasso.
\end{proof}

\medskip

The following lemma gives an implicit equation satisfied by any (non necessarily unique) minimizer $\xsoly$ of \lasso.
\begin{lem}\label{lem:sol}
 Let $\xsoly$ be a solution of $\lasso$.
 Let $I$ be the \dsup and $J$ the \dcosup of $\xsoly$ and $s = \sign(D^* \xsoly)$.
 We suppose that \eqref{eq:hj} holds.
 Then, $\xsoly$ satisfies
 \begin{equation}\label{eq:local-expression}
   \xsoly = \AJ \Phi^* y - \lambda \AJ D_I s_I .
 \end{equation}
\end{lem}
\begin{proof}
  Owing to the first order necessary and sufficient minimality condition (Lemma \ref{lem:first-order}), there exists $\sigma \in \Sig(\xsoly)$ satisfying
  \begin{equation}\label{eq:first-order-x}
    \Phi^* (\Phi \xsoly - y) + \lambda D_I s_I + \lambda D_J \sigma = 0 .
  \end{equation}
  By definition, $\xsoly \in \GJ=(\Im D_J)^\bot$.
  We can then write $\xsoly = \UJ \alpha$ for some $\alpha \in \RR^{\dim(\GJ)}$.
  Since $\UJ^* D_J = 0$, multiplying both sides of \eqref{eq:first-order-x} on the left by $\UJ^*$, we get
  \begin{equation*}
    \UJ^* \Phi^* (\Phi \UJ \alpha - y) + \lambda \UJ^* D_I s_I = 0 .
  \end{equation*}
  Since $\UJ^* \Phi^* \Phi \UJ$ is invertible, the implict equation of $\xsoly$ follows immediately.
\end{proof}

\medskip

Suppose now that a vector satisfies the above implicit equation. The next lemma derives two equivalent necessary and sufficient conditions to guarantee that this vector is actually a (possibly unique) solution to \lasso.
\begin{lem}\label{lem:soldeux}
  Let $y \in \RR^Q$ and let $J$ a \dcosup such that \eqref{eq:hj} holds, and $I = J^c$.
  Suppose that $\xsolyp$ satisfies
  \begin{equation*}
    \xsolyp = \AJ \Phi^* y  - \lambda \AJ D_I s_I .
  \end{equation*}
  where $s = \sign(D^* \xsolyp)$.
  Then, $\xsolyp$ is a solution of \lasso if, and only if, there exists $\sigma \in \RR^{|J|}$ satisfying one of the following equivalent conditions
  \begin{equation}\label{eq:cond-imp}
    \sigma - \Omega^{[J]} s_I + \frac{1}{\lambda} \BJ^{[J]} y \in \Ker D_J \qandq \normi{\sigma} \leq 1 ,
  \end{equation}
  or
  \begin{equation}\label{eq:forme-soldeux}
    \tilde \BJ^{[J]} y - \lambda \tilde \CJ^{[J]} s_I + \lambda D_J \sigma = 0 \qandq \normi{\sigma} \leq 1 ,
  \end{equation}  
  where $\tilde \CJ^{[J]} = (\Phi^* \Phi \AJ - \Id)D_I$, $\tilde \BJ^{[J]} = \Phi^* (\Phi \AJ \Phi^* - \Id)$, $\Omega^{[J]} = D_J^+ \tilde \Omega^{[J]}$ and $\BJ^{[J]} = D_J^+ \tilde \BJ^{[J]}$.\\
  Moreover, if $\normi{\sigma} < 1$ then $\xsolyp$ is the unique solution of \lasso.
\end{lem}
\begin{proof}
  First, we observe that $\xsolyp \in \GJ$.
  According to Lemma \ref{lem:first-order}, $\xsolyp$ is a solution of \lasso if, and only if, there exists $\sigma \in \Sig(\xsolyp)$.
  Since \eqref{eq:hj} holds, $\AJ$ is properly defined.
  We can then plug the assumed implicit equation in \eqref{eq:first-order-explicit} to get
  \begin{equation*}
    \Phi^* (\Phi \AJ \Phi^* y  - \lambda \Phi \AJ D_I s_I - y) + \lambda D_I s_I + \lambda D_J \sigma = 0 .
  \end{equation*}
  Rearranging the terms multiplying $y$ and $s_I$, we arrive at
  \begin{equation*}
    \Phi^*(\Phi \AJ \Phi^* - \Id) y - \lambda (\Phi^* \Phi \AJ - \Id)D_I s_I + \lambda D_J \sigma = 0 .
  \end{equation*}
  This shows that $\xsoly$ is a minimizer of \lasso if, and only if
  \begin{equation*}
    \tilde \BJ^{[J]} y - \lambda \tilde \CJ^{[J]} s_I + \lambda D_J \sigma = 0 \qandq \normi{\sigma} \leq 1 .
  \end{equation*}

  To prove the equivalence with \eqref{eq:forme-soldeux}, we first note that $\UJ^* \tilde \CJ^{[J]}=0$ implying that $\Im(\tilde \CJ^{[J]}) \subseteq \Im(D_J)$, and thus $\tilde \CJ^{[J]}=D_J \CJ^{[J]}$.
  With a similar argument, we get $\tilde \BJ^{[J]} = D_J \BJ^{[J]}$.
  Hence, the existence of $\sigma \in \Sig(\xsolyp)$ such that $\normi{\sigma} \leq 1$ is equivalent to
  \begin{equation*}
    D_J \sigma = D_J \CJ^{[J]} s_I - \frac{1}{\lambda} D_J \BJ^{[J]} y \qwhereq \normi{\sigma} \leq 1 ,
  \end{equation*}
  which in turn is equivalent to
  \begin{equation*}
    \sigma - \CJ^{[J]} s_I + \frac{1}{\lambda} \BJ^{[J]} y \in \Ker D_J \qwhereq \normi{\sigma} \leq 1 .
  \end{equation*}
  Replacing the inequality by a strict inequality condition gives the uniqueness of $\xsoly$ by virtue of Lemma~\ref{lem:first-uniqueness}.
\end{proof}

\subsection{Proof of Theorem \ref{thm:small}} 
\label{sub:small}

Recall the analysis identifiability criterion $\IC$ from Definition~\ref{def:asc}. 

\begin{proof}
  The proof is divided in three steps.
  \begin{enumerate}
    \item We give a first condition on $\lambda$ to ensure $\sign(D^* \xsolyp) = \sign(D^* x_0)$.
    \item We then derive another condition on $\frac{\norm{w}_2}{\lambda}$ to guarantee that the minimality conditions are satisfied at $\xsolyp$, and assuming $\IC < 1$ that $\xsolyp$ is the unique solution to \lasso.
    \item We finally prove that these two conditions are compatible.
  \end{enumerate}
  Let's consider the vector
  \begin{equation*}
    \xsolyp = x_0 + \AJ \Phi^* w - \lambda \AJ D_I s_I ,
  \end{equation*}
  where $s = \signxx$.
  Obviously, $\xsolyp \in \GJ$.

  \begin{enumerate}
  \item We first give a condition on $\lambda$ to ensure sign consistency, i.e.
  \begin{equation*}
    \sign(D^* \xsolyp)=\sign(D^* x_0) \eqdef s .
  \end{equation*}
  The two vectors have the same sign if
  \begin{align}\label{eq:sign-cond-small}
    \forall i \in I, \quad \abs{D_I^* x_0}_i & > \abs{D_I^* (\xsolyp - x_0)}_i \nonumber \\
    & = \abs{D_I^* \AJ \Phi^* w - \lambda D_I^* \AJ D_I s_I}_i .
  \end{align}
  Let's upper-bound $\normi{D_I^* (\xsolyp - x_0)}$ as follows
    \begin{align*}
    \normi{D_I^* (\xsolyp - x_0)} &\leq \norm{D_I^* \AJ}_{\infty, \infty} \pa{\normi{\Phi^* w} +\lambda \normi{D_I s_I} } \\
  				  &\leq \norm{D_I^* \AJ}_{\infty, \infty} \pa{\norm{\Phi^*}_{2, \infty}\norm{w}_2 + \la \norm{D_I}_{\infty, \infty}} .
    \end{align*}
  Introducing
  \begin{equation*}
    T = \umin{i \in \ens{1, \cdots, \abs{I}}} \abs{D_I^* x_0}_i > 0 ,
  \end{equation*}
  the condition
  \begin{equation}\label{eq:small-cond-un}
    T > \norm{D_I^* \AJ}_{\infty, \infty} \pa{ \norm{\Phi^*}_{2, \infty}\norm{w}_2 + \lambda \norm{D_I}_{\infty, \infty} },
  \end{equation}
  is sufficient for \eqref{eq:sign-cond-small} to hold true.

  \item We now turn to the second step of the proof.
  Observe that $\tilde \BJ^{[J]} y = \tilde \BJ^{[J]} w$ since $x_0 \in \GJ$.
  Let $\bar u \in \Ker D_J$ a minimizer of $\normi{\Omega^{[J]} s_I - u}$ over $\Ker D_J$.
  We consider the following candidate vector $\sigma \in \RR^{|J|}$ defined by
  \begin{equation*}
    \sigma = - \bar u + \Omega^{[J]} s_I - \frac{1}{\lambda} \BJ^{[J]} w .
  \end{equation*}
  We have
  \begin{equation*}
    \normi{\sigma} \leq \normi{\Omega^{[J]} s_I- \bar u} +\frac{1}{\lambda} \norm{\BJ^{[J]}}_{2, \infty} \norm{w}_2 .
  \end{equation*}
  By definition of $\bar u$,
  \begin{equation*}
    \normi{\sigma} \leq \IC(s) + \frac{1}{\lambda} \norm{\BJ^{[J]}}_{2,\infty} \norm{w}_2 .
  \end{equation*}
  Thus, since $\IC(\signxx) < 1$ and provided that
  \begin{equation}\label{eq:small-cond-deux}
    \norm{\BJ^{[J]}}_{2, \infty} \frac{\norm{w}_2}{\lambda} < 1 - \IC(\signxx) ,
  \end{equation}
  we have $\normi{\sigma} < 1$. Appealing to Lemma \ref{lem:soldeux}, it follows that $\xsolyp$ is the unique solution of \lasso.

  \item Let us show that \eqref{eq:small-cond-un} and \eqref{eq:small-cond-deux} are in agreement.
  We introduce the constants $c_J$ and $\tilde{c}_J$,
   \begin{equation*}
    c_J = \frac{\norm{\BJ^{[J]}}_{2,\infty}}{1-\IC(\signxx)} ,
  \end{equation*}
  and
  \begin{equation*}
    \tilde{c}_J = \left[\norm{D_I^* \AJ}_{\infty, \infty} \pa{ \frac{\norm{\Phi^*}_{2, \infty}}{c_J} + \norm{D_I}_{\infty, \infty} } \right]^{-1} .
  \end{equation*}
  On the one hand, if
  \begin{equation*}
    \lambda < T \tilde{c}_J ,
  \end{equation*}
  then
  \begin{equation*}
    T > \lambda \norm{D_I^* \AJ}_{\infty, \infty} \pa{ \frac{\norm{\Phi^*}_{2, \infty}}{c_J} + \norm{D_I}_{\infty, \infty} }.
  \end{equation*}
  On the other hand, if
  \begin{equation*}
    c_J \norm{w}_2 < \lambda.
  \end{equation*}
  then
  \begin{equation*}
    T > \norm{D_I^* \AJ}_{\infty, \infty} \pa{ \norm{\Phi^*}_{2, \infty} \norm{w}_2 + \lambda \norm{D_I}_{\infty, \infty} }
  \end{equation*}
  which is condition \eqref{eq:small-cond-un}.
  Moreover, $c_J \norm{w}_2 < \lambda$ also implies that
  \begin{equation*}
    \frac{\norm{\BJ^{[J]}}_{2,\infty}}{1-\IC(\signxx)} \frac{\norm{w}_2}{\lambda} < 1 ,
  \end{equation*}
  which is condition \eqref{eq:small-cond-deux}.
 \end{enumerate}
\end{proof}

\subsection{Proof of Proposition \ref{prop:icsupnoise}}
Proposition \ref{prop:icsupnoise} is a simple consequence of Lemmata~\ref{lem:sol} and~\ref{lem:soldeux}.
\begin{proof}
	Let $\xsoly$ be a solution of \lasso.
  Suppose that $\signxx = \sign(D^* \xsoly)$.
  As a consequence, $J$ is the \dcosup of $\xsoly$.
  According to Lemmata \ref{lem:sol} and \ref{lem:soldeux}, there exists $\sigma$ such that $\normi{\sigma} \leq 1$ and
    \begin{equation*}
    \sigma - \Omega^{[J]} s_I + \frac{1}{\lambda} \BJ^{[J]} w \in \Ker D_J \qwhereq s = \signxx ,
  \end{equation*}
  or equivalently, there exists $-u \in \Ker D_J$ such that
  \begin{equation*}
    \sigma = \Omega^{[J]} s_I - u - \frac{1}{\lambda} \BJ^{[J]} w .
  \end{equation*}
  It follows that
  \begin{equation*}
  	\normi{\sigma} \geq 
  		\left| \normi{\Omega^{[J]} s_I - u}
  		- \frac{1}{\lambda} \normi{\BJ^{[J]} w} \right|
  \end{equation*}
  Since $\normi{\Omega^{[J]} s_I - u} \geq \IC(s)$ and $\frac{1}{\lambda} \normi{\BJ^{[J]} w} < \IC(s) - 1$ by assumption, we have
  \begin{equation*}
  	\normi{\Omega^{[J]} s_I - u} - \frac{1}{\lambda} \normi{\BJ^{[J]} w} 
  	\geq \IC(s) - \frac{1}{\lambda} \normi{\BJ^{[J]} w} 
  	> 1
  \end{equation*}
  This imples
  \begin{equation*}
  	\normi{\sigma} > 1,
  \end{equation*}
  which is a contradiction.
\end{proof}

\subsection{Proof of Theorem \ref{thm:noiseless}} 
\label{sub:noiseless}

Theorem \ref{thm:noiseless} is proved in three steps.
\begin{enumerate}
\item First, we specialize Theorem \ref{thm:small} to the case $w=0$.
\item Then, we show that under the condition $\IC(\signxx) < 1$, the vector $x_0$ is a solution of \bp.
\item Finally, we prove Theorem \ref{thm:noiseless} by considering another feasible vector of \bp.
\end{enumerate}

\begin{cor}\label{cor:unique}
  Let $x_0 \in \RR^N$ be a fixed vector, $I$ be its \dsup, and $y = \Phi x_0$.
  Suppose that \eqref{eq:hj} holds and $\IC(\signxx) < 1$.
  Let $T = \min_{i \in \ens{1, \cdots, \abs{I}}} \abs{D_I^* x_0}_i$.
  Then for $\lambda < T \tilde{c}_J$,
  \begin{equation*}
    \xsolyp = x_0 - \lambda \AJ D_I s_I \qwhereq s = \signxx .
  \end{equation*}
  is the unique solution of \lasso.
\end{cor}
\begin{proof}
  Take $w = 0$ in Theorem~\ref{thm:small}.
\end{proof}

\begin{lem}\label{lem:cv-constrained}
  Let $x_0 \in \RR^N$ be a fixed vector, $I$ be its \dsup, and $y = \Phi x_0$.
  Suppose that \eqref{eq:hj} holds and $\IC(\signxx) < 1$.
  Then $x_0$ is a solution of \bp.
\end{lem}
\begin{proof}
  According to Corollary \ref{cor:unique}, $\lasso$ has a unique solution for $\lambda < T \tilde{c}_J$,
  \begin{equation*}
    \xsoly_\lambda \eqdef \xsolyp_\lambda = x_0 - \lambda \AJ D_I s_I,
  \end{equation*}
  where $s = \signxx$.
  Let $x_{(1)} \neq x_0$ such that $\Phi x_{(1)} = y$.
  For every $\lambda > 0$, one has $\obj(\xsoly_\lambda) < \obj(x_{(1)})$ by definition of $x_\lambda$.
  Then,
  \begin{equation*}
    \norm{D^* \xsoly_\lambda}_1 < \norm{D^* x_{(1)}}_1 .
  \end{equation*}
  By continuity of the norm, and taking the limit as $\lambda \to 0$ in the last inequality yields
  \begin{equation*}
    \norm{D^* x_0}_1 \leqslant \norm{D^* x_{(1)}}_1 ,
  \end{equation*}
  whence it follows that $x_0$ is a solution of \bp.
\end{proof}

\begin{proof}[Proof of Theorem \ref{thm:noiseless}]
  Using Lemma \ref{lem:cv-constrained}, $x_0$ is a solution of \bp.
  We shall prove that $x_0$ is actually unique.
  Let
  \begin{equation*}
    x_{(1)} = x_0 + \lambda \AJ D_I s_I .
  \end{equation*}
  For $\lambda$ small enough, one has $\sign(D^* x_{(1)}) = \sign(D^* x_0)$.
  Then if $\IC(\signxx) < 1$, it follows from Corollary \ref{cor:unique} that $x_0$ is the unique solution of (\lassoP{y_1}{\lambda}) where $y_1 = \Phi x_{(1)}$.

  Let $x_{(2)} \in \RR^N$ be another feasible point of $\bp$, i.e. $\Phi x_{(2)} = y = \Phi x_0$ with $x_{(2)} \neq x_0$.
  Since $x_0$ is the unique solution of (\lassoP{y_1}{\lambda}), we obtain
  \begin{equation*}
    \dfrac{1}{2} \norm{y_1 - \Phi x_0}_2^2 + \lambda \normu{D^* x_0} < \dfrac{1}{2} \norm{y_1 - \Phi x_{(2)}}_2^2 + \lambda \normu{D^* x_{(2)}} 
  \end{equation*}
  which implies that
  \begin{equation*}
    \normu{D^* x_0} < \normu{D^* x_{(2)}} .
  \end{equation*}
  This proves that indeed $x_0$ is the unique solution of \bp.
\end{proof}


\subsection{Proof of Theorem \ref{thm:noise}} 
\label{sub:noisy}

Recall the Recovery Criterion $\ARC$ from Definition~\ref{def:arc}. 

\begin{proof}
  Consider the following restricted problem
  \begin{equation}\label{eq:restricted-lasso}\tag{$\lassoRtag$}
    \umin{x \in \GJ} \minform .
  \end{equation}
  Our strategy is to construct a solution of \lassoR, and to show that it is the unique solution of \lasso.
  To achieve this goal, we split the proof into four steps:
  \begin{enumerate}
    \item We exhibit $p_I^\star \in \RR^{\abs{I}}$ such that
    \begin{equation*}
      \UJ^* \left[ \Phi^* ( \Phi \xsoly - y) + \lambda D_I p_I^\star \right] = 0 .
    \end{equation*}
    \item We prove that $\xsoly$ satisfies an implicit equation of the form
    \begin{equation*}
      \xsoly = \AJ \Phi^* y - \lambda \AJ D_I p_I^\star .
    \end{equation*}
    \item We prove that $\xsoly$ satisfies the first-order minimality conditions of Lemma~\ref{lem:first-order} using the construction of $p_I^\star$.
    \item Finally, we derive the $\ldeux$-robustness bound.
  \end{enumerate}
  By a simple change of variable $x=U\alpha$, we rewrite \lassoR in an unconstrained form
  \begin{equation*}
    \uargmin{\alpha \in \RR^{\dim \GJ}} \dfrac{1}{2} \norm{y - \Phi \UJ \alpha}_2^2 + \lambda \normu{D_I^* \UJ \alpha} .
  \end{equation*}

  \begin{enumerate}
  \item Applying Lemma \ref{lem:first-order} with $\Phi \UJ$ and $D_I^* \UJ$ instead of $\Phi$ and $D^*$, $\alpha^\star$ is a solution of \lassoR if, and only if, there exists $\sigma^\star$ with $\normi{\sigma^\star} \leq 1$ such that
  \begin{equation*}
    \UJ^* \Phi^* ( \Phi \UJ \alpha^\star - y) + \lambda (\UJ^* D_I)_{I^\star} s_{I^\star} + \lambda (\UJ^* D_I)_{J^\star} \sigma^\star = 0 ,
  \end{equation*}
  where $I^\star \subseteq I$ is the $(U^*D_I)$-support of $U \alpha^\star$ and $J^\star = I \setminus I^\star$.
  We introduce $p_I^\star \in \RR^{\abs{I}}$ defined as
  \begin{equation*}
    \forall i \in I, \quad
    \left( p_I^\star \right)_i
    =
    \begin{cases}
      s_i & \text{if } i \in I^\star \\
      \sigma_i^\star & \text{if } i \in J^\star,
    \end{cases}
  \end{equation*}
  which satisfies
  \begin{equation*}
    D_I p_I^\star = D_{I^\star} s_{I^\star} + D_{J^\star} \sigma^\star .
  \end{equation*}
  The above first-order optimality condition then takes the compact form
  \begin{equation}\label{eq:fst-rest}
    \UJ^* \left[ \Phi^* ( \Phi \UJ \alpha^\star - y) + \lambda D_I p_I^\star \right] = 0 .
  \end{equation}

  \item  Owing to condition \eqref{eq:hj}, $U^* \Phi^* \Phi U$ is invertible, and we obtain
  \begin{equation*}
    \alpha^\star = (U^* \Phi^* \Phi U)^{-1} U^* \Phi^* y - \lambda (U^* \Phi^* \Phi U)^{-1} U^* D_I p_I^\star .
  \end{equation*}
  Multiplying both sides by $\UJ$ recovers $\xsoly=\UJ\alpha^\star$ as
  \begin{equation}\label{eq:imp-rest}
    \xsoly = \AJ \Phi^* y - \lambda \AJ D_I p_I^\star .
  \end{equation}

  \item We now prove that $\xsoly$ is a solution of \lasso, i.e. there exists $\sigma$ such that
  \begin{equation*}
    \Phi^*(\Phi \xsoly - y) + \lambda D_{I^\star} s_{I^\star} + \lambda D_{J \cup J^\star} \sigma = 0 \qandq \normi{\sigma} \leq 1 .
  \end{equation*}
  Take $\bar u$ such that
  \begin{equation*}
    \bar u \in \uargmin{u \in \Ker D_J} \normi{\CJ^{[J]} p_I^\star - u} ,
  \end{equation*}
  and
  \begin{equation}
    \label{eq:sigbar}
    \bar \sigma = \CJ^{[J]} p_I^\star - \bar u - \frac{1}{\lambda} \BJ^{[J]} w .
  \end{equation}
  We recall from Lemma~\ref{lem:soldeux} that
  \begin{align*}
    \tilde \CJ^{[J]} &= (\Phi^* \Phi \AJ - \Id)D_I, & \tilde \BJ^{[J]} &= \Phi^* (\Phi \AJ \Phi^* - \Id), \\
    \Omega^{[J]} &= D_J^+ \tilde \Omega^{[J]}, & \BJ^{[J]} &= D_J^+ \tilde \BJ^{[J]}.
  \end{align*}
  Plugging \eqref{eq:imp-rest}, we get
  \begin{align*}
    &\Phi^*(\Phi \xsoly - y) + \lambda D_I p_I^\star + \lambda D_J \bar \sigma \\
    =\quad & \Phi^*(\Phi (\AJ \Phi^* y - \lambda \AJ D_I p_I^\star) - y) \\
    & \quad + \lambda D_I p_I^\star + \lambda D_J D_J^+ \tilde \CJ^{[J]} p_I^\star \\
    & \quad - \underbrace{\lambda D_J \bar u}_{=0} - D_J D_J^+ \tilde \BJ^{[J]} y \\
    =\quad & (\Id - D_J D_J^+) (\tilde \BJ^{[J]} y - \lambda \tilde \CJ^{[J]} p_I^\star) \\
    =\quad & (\Id - D_J D_J^+) \left[ \Phi^*(\Phi \xsoly - y) + \lambda D_I p_I^\star \right] .
  \end{align*}
  Let's denote $v = \Phi^*(\Phi \xsoly - y) + \lambda D_I p_I^\star$.
  From \eqref{eq:fst-rest}, we have $v \in \ker(U^*)=\Im(U)^\bot=\GJ^\bot$.
  Since $(\Id - D_J D_J^+)$ is the orthogonal projector on $\Im(D_J)^\bot=\GJ$, we conclude that $(\Id - D_J D_J^+)v=0$.
  It then follows that
  \begin{equation*}
    \Phi^*(\Phi \xsoly - y) + \lambda D_I p_I^\star + \lambda D_J \bar \sigma = 0 .
  \end{equation*}
  We can then write the bound
  \begin{equation*}
    \normi{\bar \sigma} \leq \normi{\Omega^{[J]} p_I^\star - \bar u} + \frac{1}{\lambda} \norm{\BJ^{[J]}}_{2,\infty} \norm{w}_2 .
  \end{equation*}
  From \eqref{eq:sigbar}, and by definition of $\bar{u}$ we get the bound
  \begin{align*}
    \normi{\bar \sigma} &\leq \umin{u \in \Ker D_J} \normi{\CJ^{[J]} p_I^\star - u} + \frac{1}{\lambda} \norm{\BJ^{[J]}}_{2,\infty} \norm{w}_2 \\
    			&\leq \ARC(I) + \frac{1}{\lambda} \norm{\BJ^{[J]}}_{2,\infty} \norm{w}_2 .
  \end{align*}
  Let $\sigma$ defined by
  \begin{equation*}
    \forall j \in \ens{1,\dots,P} \setminus I, \quad \sigma_j =
    \begin{cases}
      \sigma_j^\star & \text{if } j \in J^\star \\
      \bar{\sigma}_j & \text{if } j \in J,
    \end{cases}
  \end{equation*}
  Since by assumption $\ARC(I) < 1$ and
  \begin{equation*}
    \lambda > \norm{w}_2 \frac{c_J}{1 - \ARC(I)} \qwhereq c_J = \norm{\BJ^{[J]}}_{2,\infty},
  \end{equation*}
  we get $\normi{\bar \sigma} < 1$ and $\normi{\sigma} = \max(\normi{\bar \sigma}, \normi{\sigma^\star}) \leq 1$.
  Invoking Lemma~\ref{lem:first-order}, we conclude that $\xsoly$ is a solution of \lasso.
  Moreover, since $\normi{\bar \sigma} < 1$ and \eqref{eq:hj} holds, $\xsoly$ is the unique solution of \lasso according to Lemma~\ref{lem:first-uniqueness}.
  
  \item We now bound the $\ldeux$-distance between $x_0$ and $\xsoly$.
  \begin{equation*}
    \norm{\xsoly - x_0}_2 = \norm{\AJ \Phi^* y - \lambda \AJ D_I p_I^\star - x_0}_2 .
  \end{equation*}
  Since $x_0 \in \GJ$, we have $\AJ \Phi^* y = x_0 + \AJ \Phi^* w$.
  Consequently
  \begin{multline*}
    \norm{\xsoly - x_0}_2 =    \norm{\AJ(\Phi^* w - \lambda D_I p_I^\star)}_2 \\
     			  \leq \norm{\AJ}_{2, 2} \norm{w}_2 \left( \norm{\Phi^*}_{2, 2} + \dfrac{\rho c_J}{1 - \ARC(I)} \norm{D_I}_{2, \infty} \right) .
 \end{multline*}
  \end{enumerate}
  This concludes the proof.
\end{proof}



\section*{Conclusion} 
\label{sec:conclusion}
In this paper, we provided theoretical guarantees for accurate and robust recovery with $\lun$-analysis sparse regularization.
We derived a sufficient condition under which the $D$-support and sign of the true signal can be exactly identified in presence of a small enough noise (and a fortiori without noise). We  showed that this condition for support recovery is in some sense sharp. We proposed a stronger condition to ensure a partial support recovery for arbitrary noise if the regularization parameter is sufficiently large. As a by product, these conditions also guarantee robustness in $\ldeux$-error. Some examples were provided and discussed to illustrate our results. For discrete 1-D total variation regularization, we show that staircasing induces an instability of the $D$-support, i.e. jumps are not preserved. We believe that these contributions will allow to gain a better understanding of the behavior of sparse analysis regularizations. We would like to emphasize that a distinctive feature of our approach with respect to the literature is that we have guarantees on the robustness of the cospace associated to the true signal. This approach often has a meaningful interpretation (such as the conservation of jumps for total variation regularization). 


\section*{Acknowledgment}
We would like to thank the anonymous referees for their role in improving the original manuscript. This work was partially supported by the SIGMA-Vision ERC grant.


\appendix[Proof of Proposition \ref{prop:source}]
\label{sec:proof_source}

Let $s = \sign(D^* x_0)$ and $J$ the \dcosup of $s$.
Let $\bar u \in \Ker D_J$ such that
\begin{equation*}
  \norm{\CJ^{[J]} s_I - \bar u}_\infty = \IC(s) .
\end{equation*}

Let $\alpha$ be such that $\alpha_I=s_I$ and $\alpha_J = \CJ^{[J]} s_I - \bar u$. Since $\IC(s) < 1$, we have $\norm{\alpha}_\infty=\max(s_I,\alpha_J) \leq 1$, which shows that $\alpha \in \partial \norm{\cdot}_1(D^* x_0)$, and therefore that $D\alpha \in \norm{D^* \cdot}_1 (x_0)$.

Now, as $\CJ^{[J]} = D_J^+ \tilde \Omega^{[J]}$ and $\Im \tilde \CJ^{[J]} \subseteq \Im D_J$, we have
\begin{equation}\label{eq:mult-app}
  D_J \CJ^{[J]} = D_J D_J^+ \tilde \CJ^{[J]} = P_{\Im D_J} \tilde \CJ^{[J]} = \tilde \CJ^{[J]} ,
\end{equation}
where $P_{\Im D_J}$ is the orthogonal projection on $\Im D_J$.
Since $\bar u \in \Ker D_J$ and owing to \eqref{eq:mult-app}, we get
\begin{equation*}
  D_J \alpha_J = D_J ( \CJ^{[J]} s_I - \bar u ) = D_J \CJ^{[J]} s_I = \tilde \CJ^{[J]} s_I ~.
\end{equation*}
Using the expression of $\tilde \CJ^{[J]} = (\Phi^* \Phi \AJ - \Id)D_I$, we obtain
\begin{equation*}
  D_J \alpha_J = \Phi^* \Phi \AJ D_Is_I - D_I s_I  = \Phi^* \Phi \AJ D_I s_I - D_I \alpha_I ~.
\end{equation*}
Choosing $\eta = \Phi \AJ D_I s_I$, and since $D \alpha = D_I \alpha_I + D_J \alpha_J$, we arrive at
\begin{equation*}
  \Phi^* \eta = D \alpha ~,
\end{equation*}
or equivalently that $D\alpha \in \Im \Phi^*$. This concludes the proof.

\bibliographystyle{IEEEtran}
\bibliography{robust-analysis}

\end{document}